\documentclass[a4paper,onecolumn,11pt]{quantumarticle}
\pdfoutput=1
\usepackage[utf8]{inputenc}
\usepackage[english]{babel}
\usepackage[T1]{fontenc}
\usepackage{amsmath}
\usepackage{hyperref}

\usepackage{tikz}
\usetikzlibrary{decorations.pathreplacing,angles,quotes,calc,shapes,arrows,automata}

\usepackage{amsmath}
\usepackage{amsfonts}
\usepackage{amssymb}
\usepackage{amsthm}
\usepackage{physics}
\usepackage{bm}

\usepackage{comment}

\newcommand{\cB}{\mathcal{B}}
\newcommand{\cC}{\mathcal{C}}
\newcommand{\cD}{\mathcal{D}}
\newcommand{\cE}{\mathcal{E}}

\newcommand{\cL}{\mathcal{L}}
\newcommand{\cM}{\mathcal{M}}

\newcommand{\cP}{\mathcal{P}}

\newcommand{\cS}{\mathcal{S}}

\newcommand{\cU}{\mathcal{U}}
\newcommand{\cV}{\mathcal{V}}
\newcommand{\cW}{\mathcal{W}}

\newcommand{\bbR}{\mathbb{R}}
\newcommand{\bbI}{\mathbb{I}}
\newcommand{\bbC}{\mathbb{C}}
\newcommand{\bbN}{\mathbb{N}}

\newcommand{\heta}{\hat{\eta}}
\newcommand{\ceta}{\check{\eta}}
\newcommand{\hkappa}{\hat{\kappa}}
\newcommand{\ckappa}{\check{\kappa}}
\newcommand{\bsigma}{\bm{\sigma}}
\newcommand{\br}{\bm{r}}

\newcommand{\cptp}{\operatorname{CPTP}}
\newcommand{\hptp}{\operatorname{HPTP}}

\newcommand{\codim}{\operatorname{codim}}

\newcommand{\vspan}{\operatorname{span}}
\newcommand{\Ker}{\operatorname{Ker}}
\newcommand{\Herm}{\operatorname{Herm}}
\newcommand{\HermZero}{\Herm_0}
\newcommand{\id}{\operatorname{id}}

\newtheorem{theorem}{Theorem}
\newtheorem{lemma}[theorem]{Lemma}
\newtheorem{proposition}[theorem]{Proposition}
\newtheorem{definition}[theorem]{Definition}
\newtheorem{corollary}[theorem]{Corollary}
\newtheorem{example}[theorem]{Example}
\newtheorem{remark}[theorem]{Remark}

\begin{document}

\title{Contraction and Expansion Values of Quantum Channels}

\author{Ruben Ibarrondo}
\affiliation{Department of Physical Chemistry, University of the Basque Country UPV/EHU, Apartado 644, 48940 Leioa, Spain}
\affiliation{EHU Quantum Center, University of the Basque Country UPV/EHU, Barrio Sarriena s/n, 48940 Leioa, Spain}
\email{rubenibarrondo@gmail.com }
\orcid{0000-0002-2445-2701}
\author{Mikel Sanz}
\email{mikel.sanz@ehu.eus}
\orcid{0000-0003-0290-4698}
\affiliation{Department of Physical Chemistry, University of the Basque Country UPV/EHU, Apartado 644, 48940 Leioa, Spain}
\affiliation{EHU Quantum Center, University of the Basque Country UPV/EHU, Barrio Sarriena s/n, 48940 Leioa, Spain}
\affiliation{Basque Center for Applied Mathematics (BCAM), Alameda Mazarredo 14, 48009 Bilbao, Spain}

\maketitle

\begin{abstract}
  The contraction coefficient of the trace distance is a central tool in quantum information, quantifying how strongly a quantum channel degrades the distinguishability of states. However, being an extremal ratio, it captures only the most optimistic behaviour of the channel and is often trivial, even for very noisy channels. Moreover, a single scalar is poorly suited to describe how contraction accumulates under channel composition. In this work we introduce the \emph{contraction and expansion values}, two monotone sequences that refine the contraction and expansion coefficients in the same way singular values refine the operator norm. They arise from a min--max variational principle over subspaces of traceless Hermitian operators, admit an operational interpretation in terms of two state-discrimination games, and are shown to coincide with the Gel'fand or Bernstein numbers of the channel restricted to traceless operators. This identification places the sequences within Pietsch's theory of $s$-numbers and yields, in particular, bounds under channel composition that the contraction coefficient alone cannot provide. We establish their main structural properties and compute or estimate them for single-qubit channels, $d$-dimensional amplitude damping channels, and direct-sum channels.
\end{abstract}

Quantum channels allow us to think of quantum evolutions as input-output processes. One of their intrinsic features is that they degrade the information from their input. A paradigmatic case is that of the trace distance \cite{NielsenChuang2010,Watrous2018}, given by $D_{\Tr}(\rho \| \sigma) = \frac{1}{2} \norm{\rho - \sigma}_1$ for two quantum states $\rho$ and $\sigma$ and where $\norm{\cdot}_1$ is the Schatten $1$-norm. Here the loss of information is captured in the form of its  \emph{data processing inequality}:
\begin{equation}
    D_{\Tr}(T(\rho) \| T(\sigma)) \leq D_{\Tr}(\rho \| \sigma),
\end{equation}
where $T$ is a quantum channel. We refer to this degradation as the \emph{contraction} of the trace distance. In general, the strength of the contraction depends on the input states as well as the channel itself. In order to characterize the channel beyond specific states, the most commonly used quantifier is the contraction coefficient,
\begin{equation}
    \heta(T) := \max_{\rho\neq\sigma} \frac{D_{\Tr}(T(\rho) \| T(\sigma))}{D_{\Tr}(\rho \| \sigma)}
    = \max_{\rho\neq\sigma} \frac{\norm{T(\rho) - T(\sigma)}_1}{\norm{\rho - \sigma}_1}.
\end{equation}
When $\heta(T)<1$, we say that the channel satisfies a \emph{strong data processing inequality}. The converse quantity, which is attracting increasing attention~\cite{Ramakrishnan2021,Hirche2024,Belzig2025,Laracuente2023}, is the expansion coefficient,
\begin{equation}
    \ceta(T) := \min_{\rho\neq\sigma} \frac{D_{\Tr}(T(\rho) \| T(\sigma))}{D_{\Tr}(\rho \| \sigma)}
    = \min_{\rho\neq\sigma} \frac{\norm{T(\rho) - T(\sigma)}_1}{\norm{\rho - \sigma}_1}.
\end{equation}

Beyond foundational interest \cite{PetzRuskai1998,Hiai15,Hirche23}, contraction coefficients have been used to derive bounds on channel capacities~\cite{Hirche2022}, quantum machine learning \cite{Angrisani2023DPamplificationquantum,Angrisani2023QDP,Nuradha25}, limitations of quantum computation~\cite{StilckFrana2021,aharonov_depo}, or quantifying causal relations~\cite{Chiribella2025}. A recurrent application is the estimation of mixing times for dissipative dynamics and repeated applications of a channel~\cite{MllerHermes2018,Temme2010,gibbs_samplers}. Nevertheless, the contraction coefficient reflects an extreme-case ratio over all input states, which may not capture the typical contraction behavior of the channel. An illustrative example is the completely dephasing channel, whose contraction coefficient is one even though it completely suppresses off-diagonal coherences, suggesting a much worse typical contraction. This noticeable gap is more than a plain curiosity as it has been shown to imply stringent restrictions on error mitigation techniques \cite{Quek2024}. Furthermore, it has motivated the development of a dedicated framework to characterize the contraction of quantum channels under random distributions~\cite{Ibarrondo2025}.

Many of the applications above involve channel concatenations that result in complex evolutions or protocols, which reveals a related shortcoming. In such cases, one is interested in how contraction accumulates under composition, but considering extreme-case quantities for the components leads to an even more extreme value for the composite channel. Indeed, a scalar characterization cannot capture the possible interplay of different ``contraction directions'' which could give more informative bounds. This is what happens with the eigenvalues and singular values, which satisfy well known inequalities that account for this interplay \cite{Bhatia1996}.  When such spectral techniques apply to a linear process (be it a quantum channel or of another nature), the eigenvalues naturally provide a richer description of the convergence properties of the channel: the leading eigenvalue determines the asymptotic fixed point, while subleading ones quantify successive modes of decay.

Unfortunately, it is often useless to directly resort to the eigenvalue or singular value decomposition to study the trace distance contraction of a quantum channel. There is a solid knowledge about the spectral properties of quantum channels~\cite{Wolf2010,Wolf2012}, and it is well known that, in general, diagonalization is not guaranteed. This lack of diagonalizability hinders the direct use of spectral methods~\cite{Terhal2000}, unless additional assumptions, such as detailed balance, are imposed~\cite{Kastoryano2013,MllerHermes2018,Gao2022}. Whereas for the singular values, the problem is that they can exceed unity and even grow with the dimension~\cite{PerezGarcia2006}, which renders many bounds trivial. In Ref.~\cite{Temme2010}, this limitation was addressed by defining a modified map whose singular values are properly bounded, enabling effective bounds on convergence speed. Their analysis focuses on the leading singular values, yet it already hints at a potential hierarchy of contraction factors. However, the compositional properties could only be applied for channels with the same fixed point.

These observations suggest that an intermediate object is missing. Not a single extremal ratio, but an ordered sequence of contraction factors that resolves the channel into a hierarchy of ``contraction directions'', in the same way that the singular values resolve a matrix. Such a sequence should reduce to the contraction coefficient at one end and to the expansion coefficient at the other, retain a clear operational meaning in terms of state distinguishability, and obey composition inequalities analogous to those of the singular values, so that the contraction of a concatenated channel can be controlled by the contraction of its parts. The trace distance, however, is governed by the Schatten $1$-norm rather than an inner-product norm, so the usual singular value decomposition does not apply directly. The central question of this work is therefore whether a $1$-norm analogue of the singular values exists that meets these requirements.

In this article, we propose the \emph{contraction and expansion values}, which extend the contraction coefficient and resemble the singular values in essence. These are two sequences describing the contraction induced by a quantum channel in a spectrum-like manner, with the first element corresponding to the contraction coefficient and descending to the expansion coefficient. 
They have a clear physical interpretation and we show that they can be understood within established tools in the analysis of bounded operators.
Remarkably, they are suitable for studying channel composition because they satisfy similar inequalities to the ones for the singular values of matrices under matrix multiplication.
Like the contraction coefficient, they are generally difficult to compute exactly~\cite{Delsol2025}, but their behavior under composition suggests they can often be approximated in terms of simpler channels.

After establishing some notation in Section~\ref{sec:notation}, in Section~\ref{sec:literature_review}, we review the literature related to channel contraction. We pay special attention to the contraction coefficient, and channel composition, which is specially relevant and is one of the main applications of the tools we introduce.
In Section~\ref{sec:definition}, we introduce the contraction and expansion values.
Afterwards, in Section~\ref{sec:adversarial_game_interpretation}, we describe the adversarial game interpretation of these quantities which roots them with a clear physical meaning.
Section~\ref{sec:s_numbers} establishes the connection with $s$-numbers, which are well known tools in the contexts of bounded operator analysis. This connection allows us to prove the main properties of these two sequences, gathered in Section~\ref{sec:properties}. Finally, we consider some examples in Section~\ref{sec:examples}, which include a full characterization for single-qubit channels, estimates for the $d$-dimensional amplitude damping, and inequalities for direct sum channels in terms of its constituents.

\section{Notation}~\label{sec:notation}

In what follows, $\cM_{d,d'}$ denotes the space of $d\times d'$ complex matrices, with $\cM_d=\cM_{d,d}$. For a matrix $A$, $A^{\dag}$ denotes its Hermitian conjugate. We write $\Herm(d)$ for the real vector space of Hermitian matrices and $\HermZero(d)$ for its traceless subspace. The identity matrix is denoted by $\bbI_d$, and $\id_d: \cM_d \to \cM_d$ denotes the identity map.

Most results will be stated for \emph{Hermiticity-preserving and traceless-preserving} ($\hptp_0$) maps. A linear map $T:\cM_d\to\cM_{d'}$ is $\hptp_0$ if
\begin{equation}
    (T(A))^{\dag} = T(A^{\dag})
    \quad \text{and} \quad
    \tr(T(X)) = 0 \quad \text{whenever }\tr(X) = 0.
\end{equation}
We write $\hptp_0(d,d')$ for the set of such maps and $\hptp_0(d)$ for $d=d'$. Although quantum channels are our primary objects of interest, this larger class is technically convenient: channel differences and inverses (when they exist) are $\hptp_0$. For any $T\in\hptp_0(d,d')$, the restriction $T|_0 : \HermZero(d) \to \HermZero(d')$ is well defined. For a quantum channel $T$, $\Ker T|_0 = \Ker T$ and $\rank(T) = \rank(T|_0) + 1$.

A subspace $\cM\subseteq\HermZero(d)$ can be characterized by a maximal set of linearly independent traceless observables $\{M_1,\dots,M_n\}\subset \HermZero(d)$ orthogonal to it, such that
\begin{equation}
    \cM =\{ X\in \HermZero(d) \;\mid\; \tr(M_k X) = 0 \quad \text{for all} \quad 1\leq k\leq n \}.
\end{equation}
Hence $\codim \cM = n$ and $\dim\cM=d^2-1-n$, where the codimension is taken inside $\HermZero(d)$. The subspace $\cM$ induces a set of state pairs
\begin{equation}
    \qty{ (\rho, \sigma) \;\mid\; \rho - \sigma \in \cM }
    =
    \qty{ (\rho, \sigma) \;\mid\; \tr( M_k \rho) = \tr( M_k \sigma) \quad \text{for all} \quad 1\leq k\leq n },
\end{equation}
that is, state pairs with identical expectation value for all $M_k$. We finally define the associated unit sphere
\begin{equation}
    \cS_{\cM} =\qty{X \in \cM \;\mid\; \norm{X}_1 = 1}.
\end{equation}

\section{Extensions of the contraction coefficient and channel composition}\label{sec:literature_review}

In this section we review the literature on the contraction of quantum channels, with two goals in mind: to motivate, more technically, the characterization of channel composition and divisibility in terms of contraction, and to position the contraction and expansion values among existing refinements of the contraction coefficient.

We begin with the contraction coefficient itself. The most remarkable results include the study of equivalence relations and inequalities between contraction coefficients of different metrics \cite{PetzRuskai1998,LesniewskiRuskai1999,Sharma2012,Hiai15,Hirche23}. For instance, it is known that some families of metrics share the same contraction coefficient, and that the contraction coefficient of the trace distance dominates many others. Additionally, there are known estimations of their values for concrete metrics, predominantly for the trace distance \cite{Reeb2011, Hirche2024} and the relative entropy \cite{Berta2023}. Finally, even though matrix $p$-norms do not define monotone distances, their contraction properties have been actively studied due to their relevance in matrix analysis \cite{PerezGarcia2006,King2012,Cubitt2015}.

As recalled in the introduction, the contraction coefficient $\heta$ captures how much a channel can shrink “distances” between states, measured as the maximal ratio between output and input values of a given distance measure. The main limitation of the contraction coefficient is that it only reflects the most optimistic scenario. This is often unrepresentative and can become trivial, $\heta=1$, even for very noisy channels, such as a fully dephasing channel or a $d$-dimensional amplitude damping (as defined in Section~\ref{sec:amplitude_damping}). Despite these caveats, contraction coefficients remain widely present in the literature, and there is increasing interest in the expansion coefficient $\ceta$ \cite{Ramakrishnan2021,Hirche2024,Belzig2025,Laracuente2023}. Our goal is to provide a finer characterization, interpolating between the contraction and expansion coefficients, that can be used to understand channel compositions.

This brings us to channel composition and divisibility, the main applications of the tools we introduce. The ability to analyze a quantum channel by breaking it into smaller time slices should not be taken for granted. For instance, this is possible if we consider an isolated quantum system under a Hamiltonian $H$. The evolution from time $t_1$ to $t_2$ can be described by the composition of parts $e^{- i (t_2 - t_1) H} = e^{-i (t_2 - t) H} e^{-i (t - t_1) H}$ for any $t_1 \leq t \leq t_2$. This is not possible in general for a quantum channel $T_{t_2,t_1}$, as the evolution may not be divisible into parts as in $T_{t_2, t_1} = \Phi_{t_2,t} \circ \Phi_{t, t_1}$. If a channel $T\in\cptp(d)$ admits a decomposition $T = \Phi_1 \circ \Phi_2$ such that $\Phi_1, \Phi_2 \in \cptp(d)$ are non-unitary channels, we say that $T$ is divisible \cite{Wolf2008}. An accessible review of divisibility can be found in \cite{vomEnde2025}, and a more extensive account is provided in \cite{Chruscinski2022DynamicalMapsMarkovianRegime}. Whenever a channel can be decomposed, its analysis may be simplified by estimating the properties from its constituents. This provides both a reduction in complexity and deeper structural insights.

A related notion is that of degradability. We say that a channel $S\in\cptp(d_1,d_3)$ is a degraded version of $T\in\cptp(d_1,d_2)$ if there is a channel $\Phi\in\cptp(d_2,d_3)$ such that $S = \Phi \circ T$ \cite{Hirche2022}. The degradability relation defines a preorder on quantum channels and can be leveraged to bound quantities such as channel capacities and contraction coefficients. More precisely, if $S = \Phi \circ T$, then the contraction coefficient of $T$ provides an upper bound for that of $S$. This simple observation has proven powerful, since for some channel families $\{ T_{\lambda}\}$ the search for a valid $T$ and $\Phi$ given $S$ can be formulated as a semidefinite program (SDP) \cite{Hirche2024}. When $\{T_{\lambda}\}$ is the parametrized erasure channel, this procedure leads to the Doeblin coefficient for $S$, which has also motivated further research \cite{George2025}.

Finally, we comment on existing notions that extend the contraction coefficient or attempt to better characterize contraction under composition. Other notions provide sequences characterizing contraction that are similar in spirit to our proposal. In \cite{Burrell2009}, they discuss the compression vectors, which naturally appear in anisotropic depolarizing channels. These channels are diagonal in some operator basis formed by trace-free orthogonal operators. Such a channel compresses the Bloch sphere in $d^2-1$ directions, with the magnitude of the distortion given by the corresponding compression coefficient. As we will see in Section~\ref{sec:vs_single_qubit}, this notion coincides with the sequences introduced in this work for single-qubit channels. For our purposes, the limitation of this approach is that it does not apply to arbitrary channels, lacks universality due to its dependence on the chosen operator basis, and analyzing compositions seems to require favorable alignment of the bases involved.
In Section III.A of \cite{Temme2010}, they introduce a map denoted $Q_k$ whose singular values are guaranteed to be in the range $[0, 1]$. The map $Q_k$ appears as a similarity transformation of an ergodic quantum channel $T$ which depends on its fixed point $\sigma$ and a function $k$ of a certain class. The singular values of $Q_k$ are related to the contraction properties of the generalized $\chi^2$-divergences and provide valuable estimates for the mixing time of the channel. As an example, if we let $k(x) = \frac{1}{2}(x^{-1/2} + x^{1/2})$ we obtain the divergence $\chi^2_{1/2}(\rho,\sigma) = \tr(\rho \sigma^{-1/2} \rho \sigma^{-1/2}) - 1$ and the map $Q_{1/2} (\rho) = \sigma^{-1/2} T (\sigma^{1/2} \rho \sigma^{1/2}) \sigma^{-1/2}$. This notion could potentially be used to study the composition of channels that share a common full-rank fixed point.

Lastly, let us mention a complementary line of work~\cite{Ibarrondo2025}, where the contraction coefficient is extended in a different direction through the \emph{moments of contraction}, which interpolate between the contraction coefficient and the average-case contraction with respect to a distribution over states. Although that approach has a clearer operational interpretation and yields quantities that are easier to estimate, establishing properties for channel compositions there requires further work, whereas composition is precisely the strength of the sequences introduced here.

\section{Definition of the contraction and expansion values}\label{sec:definition}

Subspaces of traceless Hermitian matrices play a central role in the quantitative notions introduced in this work. The idea of restricting contraction coefficients to suitable subspaces already appears in the classical literature \cite{Mukhamedov2019GeneralizedDobrushin}, where it serves to capture finer structural properties of stochastic maps. In the same spirit, we introduce here subspace-restricted versions of both the contraction and the expansion coefficients, which will be key tools in our subsequent analysis.

\begin{definition}[Subspace-restricted coefficients]
    Let $T\in\hptp_0(d, d')$ and let $\cM\subseteq\HermZero(d)$ be a subspace of Hermitian traceless matrices, we define the \emph{subspace-restricted contraction} and \emph{expansion coefficients}, respectively:
    \begin{align}
        \heta(T,\cM) := \max_{X\in \cS_\cM} \norm{T(X)}_1,\label{eq:heta_def} \\
        \ceta(T,\cM) := \min_{X\in \cS_\cM} \norm{T(X)}_1.\label{eq:ceta_def}
    \end{align}
\end{definition}

We get the usual contraction and expansion coefficients as special cases
\begin{equation}
    \heta(T,\HermZero(d)) = \heta(T),
    \quad \text{and} \quad
    \ceta(T,\HermZero(d)) = \ceta(T),\label{eq:etaM_to_eta}
\end{equation}
and using $\{\rho-\sigma\}$ as a shorthand for the one-dimensional space spanned by the non-trivial difference between $\rho\neq\sigma$
\begin{equation}
    \heta(T,\{\rho -\sigma\}) = \ceta(T,\{\rho -\sigma\}) = \frac{\norm{T(\rho) -T(\sigma)}_1}{\norm{\rho -\sigma}_1}. \label{eq:eta_single}
\end{equation}
Additionally, it is clear that for any $X\in\cM$,
\begin{equation}\label{eq:subspace_contraction_range}
    \ceta(T, \cM) \norm{X}_1 \leq \norm{T(X)}_1 \leq \heta(T, \cM) \norm{X}_1,
\end{equation}
and for two subspaces $\cM\subseteq \cW$ we have $\heta(T,\cM)\leq \heta(T,\cW)$ and $\ceta(T,\cW)\leq \ceta(T,\cM)$. It holds that $\ceta(ST,\cM)\geq\ceta(S,T(\cM))\ceta(T,\cM)$ for $S\in\cptp(d',d'')$, with the convention that the right-hand side vanishes when $T(\cM)=\{0\}$.

While the subspace-restricted coefficients allow for a finer analysis of the action of a channel, their values depend on the choice of subspace. Our goal is therefore to identify refinements that are intrinsic to the map itself, by selecting a preferred sequence of subspaces that yields a characterization. Inspired by the min–max characterization of singular values, we introduce two sequences of coefficients that capture the extremal contraction and expansion behavior of the channel across subspaces of increasing dimension

\begin{definition}[Contraction Values]\label{def:contraction_values}
    The \emph{contraction values} $\hkappa_n(T)$ with $n=1,\dots,d^2-1$ of a linear map $T \in\hptp_0 (d, d')$ are given by
    \begin{equation}\label{eq:def_cv}
        \hkappa_n(T) := \min_{\codim \cM=n-1} \max_{X \in \cS_{\cM}} \norm{T(X)}_1,
    \end{equation}
    where the minimization is over subspaces of Hermitian and traceless matrices $\cM\subseteq\HermZero(d)$, and the norm is the Schatten 1-norm.
\end{definition}

\begin{definition}[Expansion Values]\label{def:expansion_values}
    The \emph{expansion values} $\ckappa_n(T)$ with $n=1,\dots,d^2-1$ of a linear map $T \in \hptp_0(d,d')$ are given by
    \begin{equation}\label{eq:def_ev}
        \ckappa_n(T) := \max_{\dim \mathcal{M}=n} \min_{X \in \cS_{\cM}} \norm{T(X)}_1,
    \end{equation}
    where the maximization is over subspaces of Hermitian and traceless matrices $\cM\subseteq\HermZero(d)$, and the norm is the Schatten 1-norm.
\end{definition}

Both definitions are based on variational principles involving two competing optimizations: an outer optimization over Grassmannians, i.e. sets of subspaces of fixed dimension, and an inner one over elements within each subspace. The connection with the subspace-restricted contraction and expansion coefficients is immediate, since Eqs. \eqref{eq:def_cv} and \eqref{eq:def_ev} admit a natural reformulation in terms of these quantities,
\begin{equation}\label{eq:CVs_EVs_from_subspace_coeffs}
    \hkappa_n(T) = \min_{\codim \mathcal{M}=n-1  } \heta(T,\cM)
    \quad \text{and} \quad
    \ckappa_n(T) = \max_{\dim \mathcal{M}=n} \ceta(T,\cM),
\end{equation}
which motivates accent notation and makes explicit that
\begin{equation}
    \hkappa_1(T)=\ckappa_1(T)=\heta(T),
    \quad\text{and}\quad
    \hkappa_{d^2-1}(T)=\ckappa_{d^2-1}(T)=\ceta(T).    
\end{equation}
With these expressions in hand, it is easy to establish the following elementary relation.

\begin{lemma}\label{lemma:cvs_vs_evs}
	Let $T\in\hptp_0(d,d')$ and $1\leq n\leq d^2-1$, then
    \begin{equation}
        \ckappa_n(T) \leq \hkappa_n(T) .
    \end{equation}
\end{lemma}

\begin{proof}
    Find two subspaces $\cM,\cW\subseteq\HermZero(d)$ with $\codim\cM=n-1$ and $\dim\cW=n$, such that their subspace restricted contraction and expansion coefficients give the contraction and expansion values. A direct dimensionality argument verifies $\dim( \cM\cap \cW ) \geq 1$, then
    \begin{equation}
        \ckappa_n(T) = \ceta(T, \cW) \leq
        \ceta(T, \cM\cap\cW)\leq \heta(T, \cM\cap\cW)
        \leq \heta(T, \cM)
        = \hkappa_n(T).
    \end{equation}
\end{proof}

To highlight the connection with the standard formulation of the min-max characterization of singular values \cite{Bhatia1996}, observe that the singular values of the linear map $T$ are obtained through analogous expressions by replacing the Schatten 1-norm with the 2-norm, i.e. the Frobenius norm, both in the objective function $\norm{T(X)}_1\mapsto\norm{T(X)}_2$ and in the unit sphere constraint $\norm{X}_1=1 \mapsto \norm{X}_2=1$. It is thus reasonable to search for a connection between these sequences and the singular values. The best our efforts can show is a direct consequence of the relation between the 1-norm and the 2-norm:

\begin{proposition}[Relation with singular values]\label{prop:relation_with_singular_values}
    For $T \in \hptp_0(d, d')$ we have that
    \begin{equation}
        \frac{1}{\sqrt{d}} s_n(T|_0) \leq \ckappa_n(T)\leq \hkappa_n(T) \leq \sqrt{d'} s_n(T|_0),
    \end{equation}
    where $\hkappa_n(T),\ckappa_n(T)$ are the contraction and expansion values, and $s_n(T|_0)$ are singular values of its restriction to traceless Hermitian matrices.
\end{proposition}
\begin{proof}
    The bound follows directly from the norm inequalities $\norm{X}_2 \leq \norm{X}_1 \leq \sqrt{d}\norm{X}_2$ for $X\in\cM_{d}$ and the Courant-Fischer-Weyl principle \cite{Bhatia1996}:
    \begin{align}
        \hkappa_n(T) &:= \min_{\codim \cM=n-1} \max_{X\in \cS_\cM} \norm{T(X)}_1
                    =\min_{\codim \cM=n-1} \max_{X\in\cM, X\neq 0} \frac{\norm{T(X)}_1}{\norm{X}_1} \\
                    &\leq \min_{\codim \cM=n-1} \max_{X\in\cM, X\neq 0} \frac{\sqrt{d'}\norm{T(X)}_2}{\norm{X}_2} \\
                    &=\sqrt{d'} s_n(T|_0),
    \end{align}
    and
    \begin{align}
        \ckappa_n(T) &:= \max_{\dim \cM=n} \min_{X\in \cS_\cM} \norm{T(X)}_1
                    =\max_{\dim \cM=n} \min_{X\in\cM, X\neq 0} \frac{\norm{T(X)}_1}{\norm{X}_1} \\
                    &\geq \max_{\dim \cM=n} \min_{X\in\cM, X\neq 0} \frac{\norm{T(X)}_2}{\sqrt{d} \norm{X}_2} \\
                    &=\frac{1}{\sqrt{d}} s_n(T|_0).
    \end{align}
    Finally, use $\ckappa_n\leq \hkappa_n$.
\end{proof}

Although we cannot claim the tightness of this result, we know that the bounds generally require dimensional factors. We show an example of this in Appendix~\ref{app:example_singluar_values}, which is mainly inspired by \cite{PerezGarcia2006}.

\section{Adversarial game  interpretation}\label{sec:adversarial_game_interpretation}

We are now in a position to introduce two distinguishability games that naturally illustrate the operational interpretation of the expansion and contraction values. In both games, there are two parties, which we call Alice and Bob. Alice randomly chooses to send one of two perfectly distinguishable states $\rho$ or $\sigma$ through a quantum channel $T$, and Bob attempts to guess which of the two states was sent. Both parties know the channel and the pair of states, but only Alice knows which was sent. The constraint that the states are perfectly distinguishable implies that they must have orthogonal support, denoted $\rho \perp \sigma$. Restricting to $\rho\perp\sigma$ (orthogonal support, not necessarily pure) loses no generality because every traceless Hermitian $X$ factorizes as $X=t(\rho'-\sigma')$ with $t=\frac{1}{2}\norm{X}_1$ and $\rho',\sigma'$ the normalized positive/negative parts, which have orthogonal support; hence $\norm{T(X)}_1/\norm{X}_1=\tfrac12\norm{T(\rho')-T(\sigma')}_1$ ranges over exactly the same set. This is what makes both games faithful to Eqs.~\eqref{eq:def_cv}--\eqref{eq:def_ev}. Bob's optimal success probability is determined by the distinguishability of the received states via Helstrom bound:
\begin{equation}
P_{\text{success}}=\frac{1}{2}\left(1+\frac{1}{2} \norm{T(\rho) - T(\sigma)}_1\right).
\end{equation}
Given a channel, the game setup requires choosing the state pair. While Bob will advocate for maximizing the success probability, Alice will take a stand for minimizing it. The best strategy for Bob is to choose states that yield:
\begin{equation}
\max_{\rho\perp \sigma} P_{\text{success}}
= \frac{1}{2}\left(1+ \max_{\rho\perp \sigma} \frac{1}{2} \norm{T(\rho) - T(\sigma)}_1\right)
= \frac{1}{2}\left(1+ \heta(T)\right),
\end{equation}
whereas Alice's optimal strategy leads to:
\begin{equation}
\min_{\rho\perp \sigma} P_{\text{success}}
= \frac{1}{2}\left(1+ \min_{\rho\perp \sigma} \frac{1}{2} \norm{T(\rho) - T(\sigma)}_1\right)
= \frac{1}{2}\left(1+ \ceta(T)\right).
\end{equation}
We construct two games, each with a distinct procedure for selecting input states. In both, Bob is gradually given more freedom over Alice or vice versa, but through different constraints. These procedures generate sequences of decreasing values, interpolating between Bob’s and Alice’s optimal strategies. The contraction and expansion values emerge from the optimal strategies at each level of freedom in these games. 

\paragraph{Game for Contraction Values} The procedure is indexed by a number $n$ ranging from $n= 1$ to $n=d^2-1$, respectively corresponding to the extremes where Bob or Alice get total freedom of choice. Firstly, Alice is allowed to introduce up to $n-1$ constraints given by observables $\{M_1,\dots, M_{n-1}\}$ such that $\rho$ and $\sigma$ must have identical expectation values for all of them. Note that her choice is equivalent to selecting a subspace $\cM\subseteq\HermZero(d)$ with $\codim\cM=n-1$, and that at $n=1$ she can do nothing. Then, Bob sets $\rho$ and $\sigma$ subject to $\tr( M_k \rho) = \tr( M_k \sigma) $. Afterward, Alice sends one of the two states, chosen at random, and Bob can implement an arbitrary one-shot measurement on the received state. If both parties play their optimal strategy, Bob's success probability is given by
\begin{equation}
P^n_{\text{success}} = \frac{1}{2}\left(1 + \hkappa_n \right),
\quad \text{where} \quad
\hkappa_n = \min_{M_1,\dots,M_{n-1}} \max_{\substack{\rho \perp \sigma \\ \tr( M_k \rho) = \tr( M_k \sigma)}} \frac{1}{2}\norm{T(\rho) - T(\sigma)}_1.
\end{equation}

\paragraph{Game for Expansion Values} Similarly, we index the procedure by ranging from $n=1$ to $d^2-1$, with the same meaning as above. In this case, Bob can introduce up to $d^2-1-n$ constraints given by observables $\{M_1,\dots, M_{d^2-1-n}\}$ such that $\rho$ and $\sigma$ must have identical expectation values for all of them. Note that his choice is equivalent to selecting a subspace $\cM\subseteq\HermZero(d)$ with $\dim\cM=n$, and that at $n=1$ he fully determines the states. Then, Alice sets $\rho$ and $\sigma$ subject to $\tr( M_k\rho) = \tr( M_k \sigma)$. Afterward, Alice randomly chooses to send one of the two states, and Bob implements an arbitrary one-shot measurement on the received state. If both parties play their optimal strategy, Bob's success probability is given by
\begin{equation}
P^n_{\text{success}} = \frac{1}{2}\left(1 + \ckappa_n \right),
\quad \text{where} \quad
\ckappa_n = \max_{M_1,\dots,M_{d^2-1-n}} \min_{\substack{\rho \perp \sigma \\ \tr( M_k \rho) = \tr( M_k \sigma)}} \frac{1}{2}\norm{T(\rho) - T(\sigma)}_1.
\end{equation}

\vspace{.5cm}

\begin{figure}
    \centering
    \includegraphics[width=0.9\linewidth]{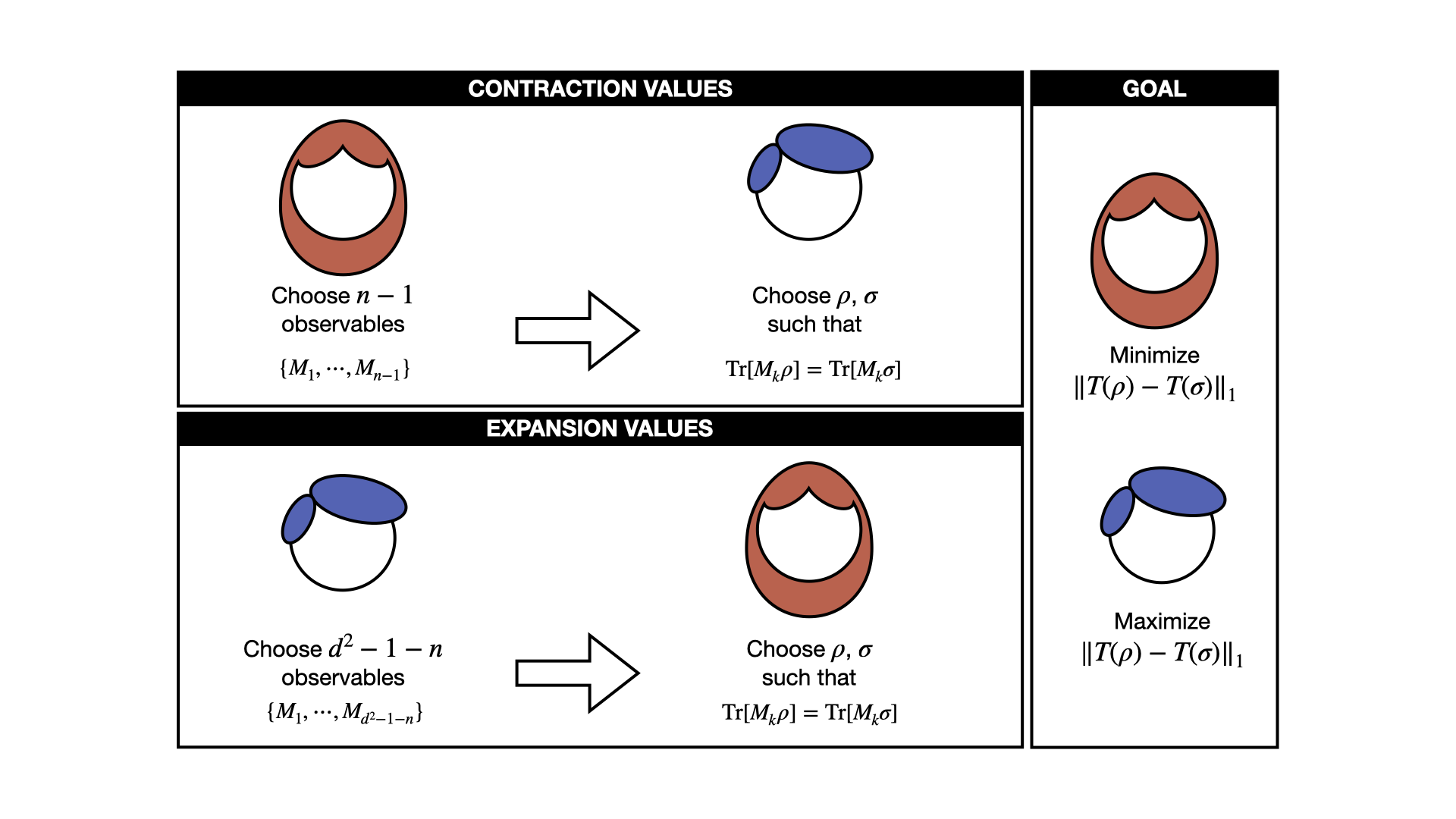}
    \caption{Given a channel $T$ and an index $1\leq n \leq d^2-1$, Alice and Bob agree on two states $\rho$ and $\sigma$ to play with. The game consists of Alice randomly choosing one of the states to send through the channel $T$, and Bob performing arbitrary one-shot measurements to guess which state was sent. The diagrams describe two ways to agree on the states, Bob’s success probability if both choose the optimal strategy is given by (a) the contraction coefficients, or (b) the expansion coefficients. }
    \label{fig:CV_EV_diagram}
\end{figure}

We represent both games diagrammatically in Fig.~\ref{fig:CV_EV_diagram}. The value sequences $\hkappa_n$ and $\ckappa_n$ resulting from the optimal strategies in the previous games are precisely the contraction and expansion values. To get an intuition for the contraction values, let us take Bob's viewpoint. He is given a quantum channel $T$ and a subspace $\cM\subseteq \HermZero(d)$ with $\codim \cM=n-1$. His goal is to choose two states $\rho$ and $\sigma$, subject to $\rho-\sigma\in\cM$, to maximize their final distinguishability after applying $T$. The $n$th contraction value, $\hkappa_n(T)$, answers the question: \textit{In the worst case, constrained to an extremely unfavorable $\cM$ with $\codim\cM=n-1$, what is the best success probability that Bob can achieve?}

On the other hand, take Bob's perspective in the second game. He is given a channel $T$ and the task of deciding a set of measurements $\{M_1,\dots, M_{d^2-1-n}\}$. His goal is to maximize the worst-case discrimination between $T(\rho)$ and $T(\sigma)$ where the input states have equal expectation values for all $M_k$. Then, the $n$th expansion value connects to the question: \textit{What is the $\cM$ with $\dim\cM=n$ that guarantees the best worst-case distinguishability?}

The restrictions to subspaces may be imposed by operational constraints, such as limited access to information, known symmetries of a physical system, or code spaces in cryptographic or error correction settings. The contraction values quantify the best achievable output distinguishability for any set of $n$ linear constraints. Whereas, in the expansion values we can view the operational constraints as tunable, and the goal is to find the $n$-dimensional subspace that guarantees the best worst-case performance. For instance, this suggests applicability to characterizing error-resilient subspaces.

\section{Connection with \texorpdfstring{$s$}{s}-numbers}\label{sec:s_numbers}

The declared motivation of this work is to generalize the singular values in such a way that we obtain a description of the contraction of the 1-norm instead of the 2-norm, preserving their relevant properties. In this Section, we show that we are indeed dealing with a proper generalization of the singular values called $s$-numbers. Precisely, we connect our definition of the contraction values to the Gel'fand numbers and the expansion values to Bernstein numbers. Below, we offer the necessary connection for readers who already know something about bounded operators, as well as for those who are simply curious about the deeper structure underlying our framework. We introduce all the required definitions as we go, although some parts may still feel a bit abstract. Even so, we encourage the reader to continue: the fact that these analytic tools turn out to fit naturally into the study of quantum channels is one of the most rewarding insights we have come across, and it places our framework on solid foundations.

Singular values give an ordered list describing how an operator stretches or contracts in different directions and how well low-rank maps can approximate it.  Many problems in analysis and numerics ask for invariants that measure ``size'', ``compactness'' or ``approximability'' of an operator; but outside Hilbert spaces, there is no canonical inner product and no unique singular-value decomposition. Pietsch's idea was to single out the minimal properties that a sequence of numbers should satisfy to be a generalization of singular values, and then study many concrete choices that fit those properties \cite{Pietsch1974}. The following introduction is based on Pietsch's work, namely his seminal paper \cite{Pietsch1974}, Chapter 11 in \cite{Pietsch1980opideals}, and Chapter 2 in \cite{Pietsch1987}.

As a short reminder, a \emph{Banach space} is a complete normed vector space. Let $E,F$ denote Banach spaces with corresponding norms $\norm{\cdot}_E,\norm{\cdot}_F$; we write $\cL (E,F)$ for the bounded linear maps $E\to F$. We endow the space of linear maps $\cL (E, F)$ with the operator norm
\begin{equation}
    \norm{T} = \sup \{ \norm{T(x)}_F :  \; x\in E, \; \norm{x}_E \leq 1 \}
    \quad \text{for} \quad
    T \in \cL (E, F).
\end{equation}
We use $J_W^E: W \to E$ to denote the injection from a subspace $W\subset E$ into $E$ acting as $J_W^E(z) = z$. Thus, the restriction of $T$ to $W$ corresponds to $T|_W = T\, J_W^E$.

\begin{definition}[$s$-numbers]\label{def:s_numbers}
    A rule $s:\; T \mapsto (s_n(T))_{n\in \bbN}$, which assigns to every bounded operator $T: E\to F$ a scalar sequence, is called an \emph{$s$-number sequence} or \emph{$s$-scale} \cite{Pietsch1987}, if the following conditions are satisfied:

    \begin{enumerate}
    \item \emph{Monotonicity.}\label{list:Monotonicity}
    \begin{equation}
        \norm{T} = s_1(T) \geq s_2(T) \geq \cdots \geq 0.
    \end{equation}
    
    \item \emph{Additivity.} \label{list:Additivity}
    For all $S,T\in\cL (E, F)$ and $n,m\in\bbN$,
    \begin{equation}
        s_{n+m-1}(S+T) \leq s_n(S) + s_m(T).
    \end{equation}
    
    \item \emph{Ideal property.} \label{list:Ideal}
    For $A\in\cL (E_0, E)$, $B\in\cL (F, F_0)$,
    \begin{equation}
        s_n(A T B) \leq \norm{A} \, s_n(T)\, \norm{B}.
    \end{equation}
    
    \item \emph{Rank property.} \label{list:Rank}
    If $\rank(T)< n$, then $s_n(T)=0$.
    
    \item \emph{Norming property.} \label{list:Norming}
    For the identity map $I_n:\ell_2^n\to \ell_2^n$, it yields $s_n(I_n) = 1$. 
    
    \end{enumerate}
\end{definition}

Many $s$-numbers satisfy an additional property:

\begin{definition}[Multiplicative $s$-numbers]
    An $s$-number sequence is said to be \emph{multiplicative} if for all Banach spaces $E,F,G$ and operators $T\in\cL (E, F)$, $S\in\cL (F, G)$, and all $n,m\in\bbN$,
    \begin{equation}\label{eq:multiplicative_s_numbers}
        s_{n+m-1}(S T) \leq s_n(S)\, s_m(T).
    \end{equation}
\end{definition}

In the case of compact operators on Hilbert spaces, Definition~\ref{def:s_numbers} completely determines the $s$-numbers, reducing to the usual singular values~\cite{Pietsch1974}, which of course are multiplicative. In general, the axioms leave room for a variety of choices. Six prominent examples are the approximation numbers, the Gel'fand numbers, the Weyl numbers, the Kolmogorov numbers, the Chang numbers, and the Hilbert numbers. We gather a brief review of their definitions and relations in Appendix~\ref{app:more_on_s_numbers}, while here we focus on the Gel'fand numbers, which are known to form a multiplicative $s$-scale.

\begin{definition}[Gel'fand numbers]
    Let  $T \in \cL(E, F)$ and $n\in \bbN$. The Gel'fand numbers are given by
    \begin{equation}    
        c_n(T) := \inf \qty{ \norm{T|_W} : W \subset E, \, \codim W < n }.
    \end{equation}
\end{definition}

To proceed with the analysis in this context, we need to identify the Banach space underlying the definition of the contraction and expansion values. This is the real vector space of traceless Hermitian operators acting on finite-dimensional Hilbert spaces equipped with the 1-norm. For a given dimension $d\in\bbN$, we identify it with $\HermZero(d)$ equipped with the 1-norm $\norm{\cdot}_1$ and denote $E_d \cong (\HermZero(d), \norm{\cdot}_1)$. 

Our interest is in $s$-numbers of operators in $\cL(E_d, E_{d'})$. In particular, we focus on operators $T|_0 \in \cL(E_d, E_{d'})$ induced by the restriction of quantum channels $T\in\cptp(d,d')$, or $\hptp_0$ maps in general. In our finite-dimensional setting, the optimizations take place over compact sets, such as unit balls or Grassmannians, and the objective functions are continuous; therefore, suprema and infima are always attained, allowing us to replace $\sup$ and $\inf$ with $\max$ and $\min$. In particular, $\norm{T|_0}=\heta(T)$.

As for any subspace $\cM\subseteq \HermZero(d)$ we have that $(T|_0)|_{\cM} = T|_{\cM}$, the Gel'fand numbers for $T|_0$ take the form
\begin{align}
    c_n(T|_0) &= \inf \qty{ \norm{T|_{\cM}} : \cM \subset E_d, \, \codim \cM < n } \\
              &= \min \qty{ \max \qty{\norm{T(X)}_1 : X \in \cS_{\cM} }: \cM \subset \HermZero(d), \, \codim \cM = n - 1} \\
              &= \hkappa_n(T).
\end{align}
This sets the following conclusion.

\begin{proposition}\label{prop:cvs_and_gelfand}
    The contraction values of a channel $T\in\cptp(d,d')$ are the Gel'fand numbers of $T|_0 \in \cL(\HermZero(d), \HermZero(d'))$.
\end{proposition}

We can similarly generalize the expansion values. For Banach spaces $E,F$ and $T\in\cL (E, F)$ we define the sequence
\begin{equation}\label{eq:bernstein_numbers}
    b_n(T) = \sup \qty{ \inf \qty{ \norm{T(x)}_F : x\in W, \, \norm{x}_E \geq 1} : W \subset E, \, \dim W \geq n}.
\end{equation}
As above, this reduces to $b_n(T|_0) = \ckappa_n(T)$ if $T$ is $\hptp_0$. Equation~\eqref{eq:bernstein_numbers} is precisely the definition of the Bernstein numbers, see Ref.~\cite{Nguyen2015} for an updated review. It is known, and it is direct to show, that $b_n$ satisfy Conditions \eqref{list:Monotonicity}, \eqref{list:Ideal}, and \eqref{list:Rank}; and the norming Condition \eqref{list:Norming}. However, they only satisfy a weaker form of additivity \eqref{list:Additivity}~\cite{Pietsch2008},
\begin{equation}\label{eq:weak_additivity_ev}
    b_n(T + S) \leq b_n(T) + \norm{S}.
\end{equation}
In fact, this was the original form proposed for the additivity property \cite{Pietsch1974}, and, in particular, it already suffices to guarantee that the sequence reduces to the usual singular values for operators on Hilbert spaces (Theorem 2.1 in \cite{Pietsch1974}). It is also established that $b_n \leq c_n$. Moreover, Ref.~\cite{Pietsch2008} reveals that the Bernstein and Gel'fand numbers generally take different values, which suggests that contraction and expansion values can be different too. However, we could not formulate an example.

\section{Properties}\label{sec:properties}

Framing the contraction and expansion values within the theory of $s$-numbers allows us to transfer structural properties directly. In particular, it guarantees that both sequences are non-increasing and equal the contraction and expansion values at the extremes. Additionally, the additivity property implies that they are continuous: small additive perturbations of the channel lead to correspondingly small variations of the associated values. While contraction and expansion values may generally differ, they nonetheless share basic qualitative features; most notably, they vanish under the same conditions. The latter is a consequence of the fourth condition for $s$-numbers, which implies that $\hkappa_n(T)=\ckappa_n(T)=0$ if and only if $n>\rank(T|_0)$. Beyond the comparison established in Lemma~\ref{lemma:cvs_vs_evs}, this perspective also yields a simple relation linking contraction and expansion values of an invertible map and its inverse, stated below.

\begin{lemma}\label{lemma:cvs_vs_evs_inv}
	Let $T\in\hptp_0(d)$ be invertible and $1\leq n \leq d^2-1$, then,
    \begin{equation}
        \hkappa_n(T) \ckappa_{d^2-n}(T^{-1}) = 1
        \quad \text{and} \quad
        \ckappa_n(T) \hkappa_{d^2-n}(T^{-1}) = 1
        .
    \end{equation}
\end{lemma}

\begin{proof}
    First, note that the contraction and expansion values of $T$ and $T^{-1}$ are strictly positive because both are of full rank. 
    Now by the absolute homogeneity of the norm, we can equivalently write
    \begin{align}
        \hkappa_n(T) &= \min_{\codim \cM=n-1} \max_{X \in \cS_{\cM}} \norm{T(X)}_1\\
                     &= \min_{\codim \cM=n-1} \max_{X \in \cM, X\neq 0} \frac{\norm{T(X)}_1}{\norm{X}_1}.
    \end{align}
    Then, as $T$ is a invertible map, we can define $Y = T(X)$ and have $X=T^{-1}(Y)$. Additionally, $T$ preserves the dimension of the subspaces, so
    \begin{align}
        \hkappa_n(T) &= \min_{\codim \cM=n-1} \max_{Y \in T(\cM), Y\neq 0} \frac{\norm{Y}_1}{\norm{T^{-1}(Y)}_1}\\
                     &= \min_{\codim \cM=n-1} \max_{Y \in \cM, Y\neq 0} \frac{\norm{Y}_1}{\norm{T^{-1}(Y)}_1}\\
                     &= \min_{\codim \cM=n-1} \max_{Y \in \cS_\cM} \frac{1}{\norm{T^{-1}(Y)}_1}\\
                     &= \qty(\max_{\codim \cM=n-1} \min_{Y \in \cS_\cM} \norm{T^{-1}(Y)}_1)^{-1}\\
                     &= \qty(\ckappa_{d^2-n}(T^{-1}))^{-1}.
    \end{align}
    This relation can also be deduced from the connection established in Section~\ref{sec:s_numbers} and Lemma 7.3 in~\cite{Pietsch1974}.
\end{proof}

It is convenient to settle that zero-padding embeddings preserve the contraction and expansion values.

\begin{lemma}\label{lemma:embedding}
    Let $d\leq D$ and $d'\leq D'$. Let $T\in\hptp_0(d,d')$, and let $\tilde{T}\in\hptp_0(D,D')$ be the map obtained by embedding $T$ into the larger spaces, namely,
    \begin{equation}
        \tilde{T}(X\oplus Y)=T(X)\oplus 0,
    \end{equation}
    where $X\in\HermZero(d)$ and $Y\in\HermZero(D)\ominus\HermZero(d)$. Then, for every $1\le n\le d^2-1$,
    \begin{equation}
        \hkappa_n(\tilde{T})=\hkappa_n(T),
        \qquad
        \ckappa_n(\tilde{T})=\ckappa_n(T).
    \end{equation}
\end{lemma}

\begin{proof}
    Write $V=\HermZero(d)$ and $\widetilde V=\HermZero(D)=V\oplus V^\perp$. For the expansion values, consider $\tilde{\cW} = \cW_1 \oplus \cW_2$ with $\cW_1 \subseteq V, \cW_2 \subseteq V^{\perp}$ and $\dim \cW = n$ for $1\leq n \leq d^2-1$. Note that if $\cW_2 \neq \{0\}$ then $\ceta(\tilde{T}, {\cW}) = 0$, whereas $\ceta(\tilde{T}, \cW_1) = \ceta(T, \cW_1) \geq 0$. Therefore,
    \begin{equation}
        \max_{\dim \tilde{\cW} = n} \ceta(\tilde{T}, \tilde{\cW}) = \max_{\dim \cW_1 = n} \ceta(T, \cW_1) = \ckappa_n(T).
    \end{equation}
    For the contraction values, consider $\tilde{\cM} = \cM_1\oplus\cM_2$ with $\cM_1 \subseteq V, \cM_2 \subseteq V^{\perp}$ and $\codim \cM = n - 1$ for $1\leq n \leq d^2-1$. As $\tilde{T}$ maps every component in $V^{\perp}$ to zero we have that $\heta(\tilde{T}, \tilde{\cM}) = \heta(T, \cM_1)$. Then, $\dim \tilde{\cM} = D^2 - n$ and $\dim \cM_2 \leq D^2  - d^2$ implies $\dim \cM_1 \geq d^2 - n$, hence as a subspace of $V$ it holds that $\codim \cM_1 \leq n - 1$. Finally,
    \begin{equation}
        \min_{\codim \tilde{\cM} = n-1} \heta(\tilde{T}, \tilde{\cM}) = \min_{\codim \cM_1 \leq n - 1} \heta(T, \cM_1) = \hkappa_n(T).
    \end{equation}
\end{proof}

The most distinctive feature of the contraction and expansion coefficients is their compositional behavior. When two channels are concatenated, the coefficients of the individual channels combine in multiple ways from `perfect alignment', where contraction effects reinforce, to `destructive misalignment', where weak contraction values of one channel counteract strong contraction values of the other. Proposition~\ref{prop:multiplicative} establishes precise bounds for the contraction and expansion values of a composite channel in terms of its constituents.

\begin{proposition}[Multiplicative properties]\label{prop:multiplicative}
	Let $T\in\hptp_0(d, d')$ and $S\in\hptp_0(d',d'')$, the following multiplicative properties hold:
    \begin{enumerate}
        \item 
        \begin{equation}
    		\hkappa_{n+m-1}(S \circ T) \leq  \hkappa_n(S) \hkappa_m(T).
    	\end{equation}
	    \item
    	\begin{equation}
    		\ckappa_{n+m-d^2+1}(S \circ T) \geq  \ckappa_n(S) \ckappa_m(T).
    	\end{equation}
	    \item
    	For $n+m=d^2$
        \begin{equation}
            \hkappa_{n}(S \circ T) \geq 
            	\max\qty{ \ckappa_{m}(S) \hkappa_n(T),
            			   \hkappa_{n}(S) \ckappa_{m}(T)},
        \end{equation}
        \item and, similarly,
        \begin{equation}
            \ckappa_{n}(S \circ T) \leq
            	\min\qty{ \ckappa_{m}(S) \hkappa_n(T),
            			  \hkappa_{n}(S) \ckappa_{m}(T)}.
        \end{equation}
    \end{enumerate}
\end{proposition}

\begin{remark}
    Note that the inequalities in Proposition~\ref{prop:multiplicative} can be equivalently stated by minimizing or maximizing over all possible indices, for instance, 
    \begin{equation}
        \hkappa_n(S \circ T) \leq \min_{1\leq k \leq n} \; \hkappa_{k}(S) \hkappa_{n-k+1}(T).
    \end{equation}
\end{remark}

\begin{proof}
    The first result is a direct consequence of the submultiplicativity of the Gel'fand numbers and identifying the contraction values with them, see Proposition~\ref{prop:cvs_and_gelfand}. For the third inequality, find subspaces $\cM\subseteq\HermZero(d)$ and $\cW\subseteq\HermZero(d')$, such that
    \begin{align}
        \ceta(S,\cW) = \ckappa_{d^2-n}(S) \quad \text{and} \quad \codim \cW = n-1, \\
        T(\cM) \subseteq \cW \quad \text{and} \quad \codim \cM \leq n-1.
    \end{align}
    Then,
    \begin{align}
        \heta(S T, \cW) &\geq \ceta(S, T(\cM)) \heta(T, \cM) \geq \ceta(S, \cW) \heta(T, \cM)\\
        &= \ckappa_{d^2-n}(S) \heta(T, \cM)\\
        &\geq \ckappa_{d^2-n}(S) \hkappa_n(T).
    \end{align}
    Similarly, find $\cM$ such that $\ceta(T, \cM) = \ckappa_{d^2-n}(T)$ with $\codim \cM = n-1$, then,
    \begin{align}
        \heta(S T, \cM) &\geq \heta(S, T(\cM)) \ceta(T, \cM)\\
        &= \heta(S, T(\cM)) \ckappa_{d^2-n}(T)\\
        &\geq \hkappa_{n}(S) \ckappa_{d^2-n}(T).
    \end{align}

    Finally, the second and fourth are directly implied by the previous two through Lemma~\ref{lemma:cvs_vs_evs_inv} for invertible maps, and using continuity to extend it to all maps. Assume first that $d=d'=d''$ and that $T$ and $S$ are invertible. Applying part (1) to $T^{-1} \circ S^{-1}$ gives
    \begin{equation}
        \hkappa_{n+m-1}(T^{-1} \circ S^{-1}) \leq  \hkappa_n(T^{-1}) \hkappa_m(S^{-1}).
    \end{equation}
    Using Lemma~\ref{lemma:cvs_vs_evs_inv},
    \begin{equation}
        \ckappa_{d^2 -(n+m-1)}(S \circ T) \geq  \ckappa_{d^2-n}(T) \ckappa_{d^2-m}(S).
    \end{equation}
    Relabeling the indices yields the desired inequality. The fourth inequality follows analogously from part (iii). Finally, since invertible maps are dense in $\hptp_0(d)$ and the contraction and expansion values depend continuously on the map, both inequalities extend to non-invertible $T$ and $S$. Finally, to extend to arbitrary dimensions $d\neq d' \neq d''$, we can embed $T$ and $S$ into a common space associated to $D=\max\{d,d',d''\}$, as shown in Lemma~\ref{lemma:embedding}, which preserves the first contraction and expansion values.
\end{proof}

As a simple instance, if we focus on the contraction and expansion coefficient, the first two inequalities combined with $\ckappa_{n}(T) \leq \hkappa_{n}(T)$ yield the intervals
\begin{align}
    \max_{1\leq k \leq d^2-1} \; \ckappa_{k}(S) \ckappa_{d^2-k}(T) 
    \leq
    \heta(S \circ T)
    \leq  \heta(S) \heta(T)
    \quad \text{and} \\
    \ceta(S) \ceta(T)
    \leq
    \ceta(S \circ T)
    \leq  
    \min_{1\leq k \leq d^2-1} \; \hkappa_{k}(S) \hkappa_{d^2-k}(T).
\end{align}

Another natural way to combine two channels is via a convex mixture. It is intriguing to ask whether additive counterparts exist for the multiplicative properties. From the previous section, the following property follows immediately.

\begin{proposition}[Additive properties]\label{prop:additive}
	Let $T, S\in\hptp_0(d,d')$ and $\lambda \in [0, 1]$. Then, the following additive property holds:
    \begin{equation}
        \hkappa_{n+m-1}(\lambda T + (1- \lambda) S) \leq  \lambda \hkappa_n(T) + (1-\lambda) \hkappa_m(S).
    \end{equation}
\end{proposition}

\begin{proof}
    This is direct from the sub-additivity of the Gel'fand numbers and identifying the contraction values with them, see Proposition~\ref{prop:cvs_and_gelfand}.
\end{proof}

In a similar vein, one can attempt to derive lower bounds, but this requires passing through the reverse triangle inequality, $\norm{X-Y}\geq\abs{\norm{X} - \norm{Y}}$, which makes the resulting estimate rather crude. Consequently, the bounds obtained are only informative for small values of $\lambda$. While sharper bounds may well exist, this already indicates that the contraction and expansion values are naturally suited to characterizing channel compositions, but are less well adapted to channel mixtures.

\section{Examples}\label{sec:examples}

Here we restrict attention to maps that are $\cptp$. Under this assumption, the simplest examples are unitary conjugations, $\cU(\rho)=U\rho U^{\dag}$. Such maps are isometries and therefore preserve the trace norm, which implies $\hkappa_n(\cU)=\ckappa_n(\cU) = 1$ for all $1\leq n \leq d^2-1$. Moreover, any $\phi\in\cptp(d)$ for which all these values are equal to $1$ must itself be a unitary conjugation. This follows from Helstrom’s characterization of the trace norm, which states that two states are at trace distance $1$ if and only if their supports are orthogonal, together with Wigner’s theorem, which asserts that any map sending pure states to pure states while preserving orthogonality is either of the form $U\rho U^{\dag}$ or $U\rho^{\top}U^{\dag}$, where $^{\top}$ denotes transposition in some fixed basis. Therefore, we get that only the unitary conjugation is a valid choice, as the transposition map is positive but not completely positive.

At the opposite extreme, $\cptp$ maps for which all contraction and expansion values vanish correspond to replacer channels of the form $\rho\mapsto \tr(\rho)\,\sigma$. Other basic examples include idempotent channels, whose contraction and expansion values coincide and are necessarily either $0$ or $1$, as well as the depolarizing channel $\cD_p(\rho)=(1-p)\rho + p \tr(\rho) \, \sigma$, for which one immediately finds $\hkappa_n(\cD_p)=\ckappa_n(\cD_p)=1-p$.  In the remainder of this section, we turn to less trivial cases, focusing on single-qubit channels, $d$-dimensional amplitude damping channels, and direct-sum channels.

\subsection{Single qubit channel} \label{sec:vs_single_qubit}

Single-qubit channels can be conveniently described in the Pauli basis
\begin{equation}\label{eq:channel_pauli_basis}
    \hat{T} = \mqty(1 & 0\\ \bm{\lambda} & \Delta)
\end{equation}
where $\Delta_{ij}=\frac{1}{2}\tr(\sigma_i T(\sigma_j))$ and $\lambda_j = \frac{1}{2} \tr(\sigma_j T(\bbI))$, with $\sigma_1 = \ketbra{0}{1}+\ketbra{1}{0}$, $\sigma_2 = -i \ketbra{0}{1}+ i\ketbra{1}{0}$, and $\sigma_3 = \ketbra{0}-\ketbra{1}$. This representation and related geometrical properties are well understood \cite{Ruskai2002}. Describing the states on this basis $\rho=\frac{1}{2}\qty(\bbI+\br\cdot \bsigma)$ where $\br\in \bbR^3$ and $\abs{\br} \leq 1$ with $\abs{\br}$ is the Euclidean norm of $\br$, and $\br\cdot\bsigma=r_1\sigma_1+r_2\sigma_2+r_3\sigma_3$, the output state is given by
\begin{equation}
    T(\rho) = \frac{1}{2}\qty(\bbI+(\bm{\lambda} + \Delta \br) \cdot \bsigma).
\end{equation}
The Bloch sphere corresponds to the set of pure states and can be identified with the unit sphere $B = \{\br \,: \, \abs{\br} = 1\} \subset \bbR^3$. In turn, the image of the Bloch sphere corresponds to an ellipsoid. Using the singular value decomposition $\Delta = \sum_{i=1}^3 s_i \mathbf{u}_i \mathbf{v}_i^{\dag}$ and introducing $\bm{y}\in\bbR^3$ with $y_i = \mathbf{v}_i^\dag \cdot \br$ (note that $\abs{\mathbf{y}}=\abs{\br}$),
\begin{equation}\label{eq:bloch_sphere_to_ellipsoid}
    T(B) = \{\bm{\lambda} + \sum_{i=1}^3 s_i y_i \mathbf{u}_i \, : \, \abs{\bm{y}} = 1 \},
\end{equation}
which is precisely an ellipsoid shifted from the origin and with principal axes given by $s_1$, $s_2$, and $s_3$.

Further, note that $\Delta$ is a matrix representation of $T|_0$ and that any traceless matrix $X\in\HermZero(2)$ admits the form
\begin{equation}
    X = \frac{1}{2} \br \cdot \bsigma
    \quad \text{so} \quad
    \norm{X}_1 = \abs{\br}.
\end{equation}
Thus, for a qubit, the 1-norm is a Euclidean norm. In terms of Section~\ref{sec:s_numbers}, this means that $E_2$ is a Hilbert space, and hence that any $s$-scale is given by the singular values of $T|_0$, therefore:

\begin{proposition}
    Let $T\in\cptp(2)$ be a single-qubit channel, then
    \begin{equation}
        \hkappa_n(T) = \ckappa_n(T) = s_n(T|_0).
    \end{equation}    
\end{proposition}
That is, the contraction and expansion values correspond to the principal axes of the ellipsoid $T(B)$ as represented in Figure~\ref{fig:bloch_sphere}.

Additionally, the channel composition $S\circ T$ can be described as
\begin{equation}
    \hat{S}\hat{T}=\mqty(1 & 0\\ \bm{\lambda}_S+\Delta_T \bm{\lambda}_T & \Delta_S\Delta_T),
\end{equation}
which shows that the multiplicativity results in Proposition~\ref{prop:multiplicative} correspond to the Gel'fand-Naimark inequality for singular values, which letting $\sigma_i = s_i(S)$ and $\tau_i = s_i(T)$ yield:
\begin{align}
    \max \{\sigma_3 \tau_1, \sigma_2 \tau_2, \sigma_1 \tau_3 \}   &\leq s_1(S T)  \leq  \sigma_1 \tau_1 ,
    \\
    \max \{ \sigma_3 \tau_2, \sigma_2 \tau_3 \} &\leq s_2(S T)  \leq \min \{ \sigma_1 \tau_2, \sigma_2 \tau_1 \} ,
    \\
    \sigma_3 \tau_3 &\leq s_3(S T)  \leq \min \{\sigma_3 \tau_1, \sigma_2 \tau_2, \sigma_1 \tau_3 \}.
\end{align}

\begin{figure}
    \centering
    \includegraphics[width=0.75\linewidth]{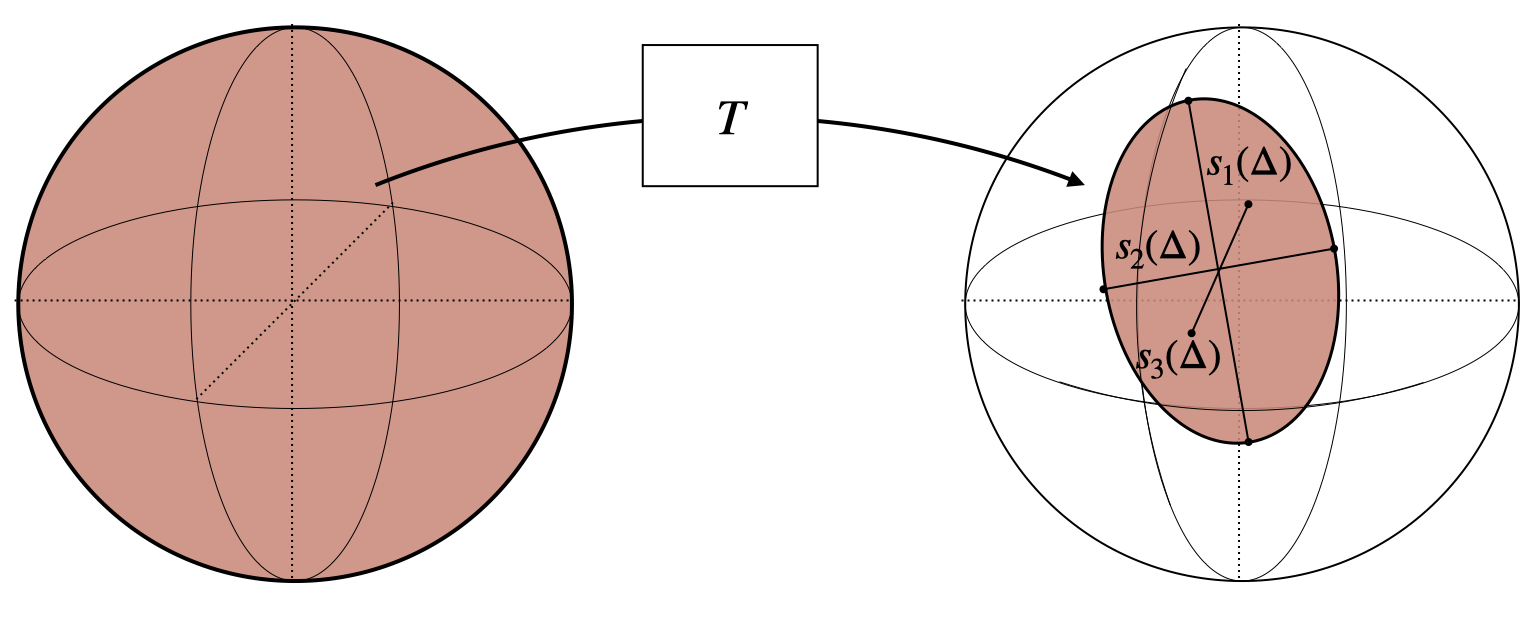}
    \caption{Under the action of a quantum channel $T$, the Bloch sphere $B$ is mapped into an ellipsoid $T(B)$, see Eq.~\eqref{eq:bloch_sphere_to_ellipsoid}, contained inside the original sphere. The principal axes of the ellipsoid correspond to the contraction and expansion values of the channel.}
    \label{fig:bloch_sphere}
\end{figure}

\subsection{Amplitude damping channel}
\label{sec:amplitude_damping}

Let $\ket{x},\ket{y}\in\bbC^d$ be two different vectors and $\lambda\in (0, 1)$, we define the $d$-dimensional $\ket{x}$ to $\ket{y}$ amplitude damping channel $\Psi_{\lambda,x,y}\in\cptp(d)$ as
\begin{gather}\label{eq:ndim_adamp}
    \Psi_{\lambda,x,y}(\rho)=K_{\lambda} \rho K^{\dag}_{\lambda}+L_{\lambda} \rho L^{\dag}_{\lambda},\\
    K_{\lambda} = \mathbb{I}-(1-\sqrt{1-\lambda})\ketbra{x},
    \quad
    L_{\lambda} = \sqrt{\lambda} \ketbra{y}{x}.
\end{gather}
The single-qubit amplitude damping is ubiquitous in the theory of open quantum systems, and its $d$-dimensional generalization provides a simple instance of a physically motivated non-unital channel. Additionally, it has recently been established a method to efficiently find whether a channel $\Phi\in\cptp(d)$ is divisible into two quantum channels as $\Phi=(\Phi \circ \Psi_{\lambda,x,y}^{-1})\circ \Psi_{\lambda,x,y}$ with $\Phi \circ \Psi_{\lambda,x,y}^{-1} \in \cptp(d)$ \cite{vomEnde2025}. Therefore, by Proposition~\ref{prop:multiplicative}, estimating the contraction and expansion values of this family immediately yields bounds for any channel $\Phi$ satisfying this divisibility criterion.

In what follows, we concentrate on the case where $x$ and $y$ are orthogonal and without loss of generality set $\ket{x}=\ket{1},\ket{y}=\ket{0}$, and denote $\Psi_{\lambda}$. Introducing the orthonormal basis $\{\ket{i}\}_{i=0}^{d-1}$, we get the following element-wise description of the channel
\begin{gather}
    \Psi_{\lambda}(\ketbra{i}{j})=\ketbra{i}{j}, \quad
    \Psi_{\lambda}(\ketbra{1})=(1-\lambda)\ketbra{1}+\lambda \ketbra{0}\\
    \Psi_{\lambda}(\ketbra{1}{i})=\sqrt{1-\lambda}\ketbra{1}{i},\quad
    \Psi_{\lambda}(\ketbra{i}{1})=\sqrt{1-\lambda}\ketbra{i}{1},
\end{gather}
where $i,j\neq 1$.

It is worth inspecting the single-qubit case first. Following the notation introduced in Section~\ref{sec:vs_single_qubit}, we describe the channel's transition matrix in the Pauli basis
\begin{equation}
    \hat{\Psi}_{\lambda} =
    \begin{pmatrix}
        1 & 0 & 0 & 0 \\
        0 & \sqrt{1-\lambda} & 0 & 0\\
        0 & 0 & \sqrt{1-\lambda} & 0\\
        \lambda & 0 & 0 & 1-\lambda
    \end{pmatrix}.
\end{equation}
We can directly read the singular values of the traceless part and determine that
\begin{equation}
    \kappa_1(\Psi_{\lambda}) = \sqrt{1-\lambda},
    \quad
    \kappa_2(\Psi_{\lambda}) = \sqrt{1-\lambda},
    \quad \text{and} \quad
    \kappa_3(\Psi_{\lambda}) = 1-\lambda,
\end{equation}
where $\kappa_n=\hkappa_n=\ckappa_n$ are the contraction and expansion values, which coincide for single-qubit channels. For the $d$-dimensional case, we arrive at the following result.

\begin{proposition}\label{prop:cev_n_dim_amplitude_damping}
    Let $\Psi_{\lambda} \in \cptp(d)$ be the $d$-dimensional amplitude damping channel defined above. We have that
    \begin{enumerate}
        \item for $1\leq n \leq (d-1)^2 - 1$,
        \begin{equation}
            \hkappa_n(\Psi_\lambda) = \ckappa_n(\Psi_\lambda) = 1,
        \end{equation}
        \item for $(d-1)^2 \leq n \leq d^2 - 2$,
        \begin{equation}
            \ckappa_n(\Psi_\lambda) \leq \hkappa_n(\Psi_\lambda) \leq \sqrt{1 - \lambda},
        \end{equation}
        \item and, finally,
        \begin{equation}
            \hkappa_{d^2-1}(\Psi_\lambda) = \ckappa_{d^2-1}(\Psi_\lambda) = 1 - \lambda.
        \end{equation}
    \end{enumerate}
\end{proposition}

\begin{proof}
   We first establish some facts about restrictions of the channel to some relevant subspaces. Let $V=\vspan\{\ket{k} \, : \, k\neq 1\}$ and observe that the space $\cB(V)$ is invariant under the action of the channel, hence $\Psi_{\lambda}|_{\cB(V)}=\text{id}_{\cB(V)}$. Note that the dimension of the traceless Hermitian matrices acting on $V$ is $\dim \HermZero(V) = (d-1)^2 - 1$. Similarly, consider $C = \qty{\ketbra{1}{\psi}+\ketbra{\psi}{1}\, :\, \ket{\psi} \in V }$, and observe that $\eval{\Psi_{\lambda}}_{C}=\sqrt{1-\lambda}\, \text{id}_{C}$. Note that as a real space $\dim C = 2(d-1)$.
    
    For the range $1\leq n\leq (d-1)^2-1$ and $\cW$ with $\dim \cW = n \leq (d-1)^2-1$, we can take $\cW\subset \HermZero(V)$ and find that $\ckappa_n(\Psi_{\lambda}) \geq 1$. We get the equalities from the standard fact $\heta(\Psi_{\lambda}) \leq 1$ and the relation $\ckappa_n\leq \hkappa_n$.
    
    For the values $(d-1)^2\leq n\leq d^2-2$ and $\cM$ with $\codim \cM \geq (d-1)^2$, we can take $\cM\subset C$ as $\dim\cM \leq d^2-1-(d-1)^2 = 2(d - 1)$. This leads to the upper bound for the contraction values $\hkappa_n(\Psi_{\lambda})\leq \sqrt{1-\lambda}$ (and hence $\ckappa_n\le\hkappa_n\le\sqrt{1-\lambda}$ by Lemma~\ref{lemma:cvs_vs_evs}).
        
    Finally, for $n = d^2-1$, we can use $X = \frac{1}{2}(\ketbra{0}-\ketbra{1})$ to prove that $\hkappa_{d^2-1}(\Psi_{\lambda})\leq 1-\lambda$. Conversely, Lemma~\ref{lemma:degrading_dephasing} shows that the amplitude damping channel can be degraded to a depolarizing channel with parameter $1-\lambda$, implying that $\ckappa_{d^2-1}(\Psi_{\lambda}) \geq  1-\lambda$.
\end{proof}

This example illustrates a method for estimating the contraction and expansion values of a family of channels whose structure in the $d$-dimensional case can be related to that of the two-dimensional one. Nonetheless, there remains room to tighten the interval $[1-\lambda,\sqrt{1-\lambda}]$ for the values in the intermediate regime. Our current understanding is consistent with a wide range of behaviors: these values could all coincide, they could be distinct, they could concentrate near one endpoint of the interval or the other, or they could be distributed more uniformly throughout it. Moreover, in this regime, the contraction and expansion values could differ; however, we have not been able to establish this even in the case $d=3$.

\subsection{Direct sum channels}\label{sec:direct_sum_channels}

Let $A\in\cM_{d_1}$ and $B\in\cM_{d_2}$ be two matrices, then their direct sum $A\oplus B \in \cM_{d_1+d_2}$ is the block-wise diagonal matrix given by
\begin{equation}
    A \oplus B = \mqty( A & 0 \\ 0 & B).
\end{equation}
Being able to decompose an operator as a direct sum greatly simplifies its analysis, for instance, $\norm{A \oplus B}_1 = \norm{A}_1 + \norm{B}_1$. Given $T_1\in\cptp(d_1,d_1')$ and $T_2\in\cptp(d_2,d_2')$, we define the direct sum channel $T_1 \oplus T_2 \in\cptp(d_1+d_2, d_1'+d_2')$ as
\begin{equation}
    (T_1\oplus T_2)\left(
    \begin{pmatrix}
        \rho_{11} & \rho_{12} \\
        \rho_{21} & \rho_{22}
    \end{pmatrix}
    \right)
    =
    \begin{pmatrix}
        T_1(\rho_{11}) & 0 \\
        0 & T_2(\rho_{22}),
    \end{pmatrix}
\end{equation}
that is, $(T_1\oplus T_2) (\rho) = T_1(\rho_{11}) \oplus T_2(\rho_{22})$. We concentrate on two terms, but both notions generalize to multiple terms as $(\oplus_{i=1}^N T_i)(\rho)= \oplus_{i=1}^N T_i(\rho_{ii})$. Lastly, we use the direct sum notation for subspaces $\cV_1 \subseteq \cM_{d_1}$ and $\cV_2 \subseteq \cM_{d_2}$ to denote $\cV_1 \oplus \cV_2 = \{x \oplus y \, : \, x\in\cV_1, y\in\cV_2 \} \subseteq \cM_{d_1+d_2}$.

Quantum states and channels with direct sum structure have been used to simplify additivity problems in quantum information quantities \cite{Shor2001, Fukuda2007}. Moreover, direct sum channels have been generalized to include some coherence in the off-diagonal blocks, enabling the extension of some of these results \cite{Chessa2021}. Alternatively, in \cite{Lonigro2023}, they introduce a generalization of amplitude damping that is suitably described in this block-wise structure, in this case including leakage from one of the diagonal blocks to the other. Recently, direct sum channels have appeared in the context of mixed unitary channels, providing an example where the Kraus rank of a mixed unitary channel is strictly smaller than the minimum number of unitaries required to describe the channel \cite{Girard2022}. These results encourage estimating the contraction and expansion values of direct sum channels in terms of their constituents. Beyond the interest in their simple structure, their application and generalization reveal promising directions for extending this preliminary analysis.

It is direct from trace preservation that the contraction coefficient of a direct sum channel is equal to one. Indeed, let $\rho_1\in\cM_{d_1}$ and $\rho_2\in\cM_{d_2}$ be two arbitrary quantum states, and define $X = \frac{1}{2} (\rho_1 \oplus (-\rho_2))$, clearly $X\in\HermZero(d_1+d_2)$ and $\norm{X}_1 = 1$. Moreover,
\begin{equation}
    \norm{(T_1 \oplus T_2)(X)}_1 = \frac{1}{2} (\norm{T_1(\rho_1)}_1 + \norm{T_2(\rho_2)}_1) = 1,
\end{equation}
thus $\heta(T_1 \oplus T_2) = 1$. Moreover, taking $\rho_1$ and $\rho_2$ to be fixed points of $T_1$ and $T_2$, respectively, leads to $T(X) = X$. Alternatively, define $\cC\subset\HermZero(d_1+d_2)$ as the subspace
\begin{equation}
    \cC = \qty{ \mqty(0 & X \\ X^\dag & 0) \,:\, X\in\cM_{d_1,d_2} },
\end{equation}
observe that $(T_1 \oplus T_2) (\xi) = 0$ for $\xi\in \cC$ and $\dim \cC = 2 d_1 d_2$, equivalently $\codim \cC = d_1^2 + d_2^2 - 1$. Thus, $\hkappa_n(T_1\oplus T_2) = \ckappa_n(T_1\oplus T_2) = 0$ for $n\geq d_1^2 + d_2^2$. For the intermediate values, we need to establish the following claim.

Define the sequence $(u_1, u_2, \dots, u_{d_1^2 + d_2^2 - 1})$ as the non-increasingly sorted concatenation of the contraction values of $T_1$ and $T_2$, and $(l_1, l_2, \dots, l_{d_1^2 + d_2^2 - 1})$ as the non-increasingly sorted concatenation of the expansion values of $T_1$ and $T_2$. These sequences serve as estimates for the contraction and expansion values of the direct sum channel $T_1\oplus T_2$.

\begin{proposition}\label{prop:direct_sum_kappas}
    With the notation introduced above, it holds that:
    \begin{equation}
        \hkappa_1(T_1\oplus T_2) = \ckappa_1(T_1\oplus T_2) = 1,
    \end{equation}
    for $2 \leq n \leq d_1^2 + d_2^2 - 2$,
    \begin{equation}\label{eq:direct_sum_kappas_intermediate}
        l_{n} \leq \ckappa_{n}(T_1\oplus T_2) \leq \hkappa_{n}(T_1\oplus T_2) \leq u_{n-1},
    \end{equation}
    for $n = d_1^2 + d_2^2 - 1$,
    \begin{equation}
        0 \leq \ckappa_{d_1^2 + d_2^2 - 1}(T_1\oplus T_2) \leq \hkappa_{d_1^2 + d_2^2 - 1}(T_1\oplus T_2) \leq u_{d_1^2 + d_2^2 - 2},
    \end{equation}
    and, for $n\geq d_1^2 + d_2^2$,
    \begin{equation}
        \hkappa_n(T_1\oplus T_2) = \ckappa_n(T_1\oplus T_2) = 0.
    \end{equation}
\end{proposition}
\begin{proof}
    See Appendix~\ref{sec:supp_mat_direct_sum}
\end{proof}

The following is a simple application of this result.

\begin{corollary}
    Let $\cD_p\in\cptp(d_1)$ and $\cD_q\in\cptp(d_2)$ be two depolarizing channels with $0 \leq  p \leq q \leq 1$. Then, for $2 \leq n \leq d_1^2$:
    \begin{equation}
        \ckappa_n(\cD_p\oplus\cD_q) = \hkappa_n(\cD_p\oplus\cD_q) = 1 - p,
    \end{equation}
    and for $d_1^2 + 1 \leq n \leq d_1^2 + d_2^2 - 2$,
    \begin{equation}
        \ckappa_n(\cD_p\oplus\cD_q) = \hkappa_n(\cD_p\oplus\cD_q) = 1 - q.
    \end{equation}
\end{corollary}

\begin{proof}
    First, note that the contraction and expansion values of the depolarizing channel $\cD_p$ are all equal to $1-p$. Concatenating and sorting them, we find that
    \begin{equation}
        u_n = l_n =
        \begin{cases}
            1 - p & 1 \leq n \leq d_1^2 - 1, \\
            1 - q & d_1^2 \leq n \leq d_1^2 + d_2^2 - 1.
        \end{cases}
    \end{equation}
    Substituting in Proposition~\ref{prop:direct_sum_kappas} leads to the result for $n\neq d_1^2$, and to the upper bound $\ckappa_{d_1^2}(\cD_p\oplus\cD_q) \leq \hkappa_{d_1^2}(\cD_p\oplus\cD_q) \leq 1 - p$. In this case, we can show that this bound is saturated. Consider the subspace $\cW$ formed by operators of the form
    \begin{equation}
        X = \omega \left(\frac{\Pi_{1}}{d_1} - \frac{\Pi_{2}}{d_2}\right) + \Delta 
        \quad \text{with} \; \omega\in\bbR, \; \Delta \in \HermZero(d_1)\oplus 0,
    \end{equation}
    where $\Pi_{1}$ and $\Pi_{2}$ are the orthogonal projectors on the first and second subspaces. Observe that $\dim\cW = d_1^2$, so it suffices to show that $\ceta(\cD_p\oplus\cD_q, \cW)\geq 1 - p$. Let us set
    \begin{equation}
        \norm{X}_1 = \omega + \norm{\Delta + \omega \frac{\Pi_{1}}{d_1}}_1 = 1,
    \end{equation}
    then
    \begin{align}
        \norm{\cD_p\oplus\cD_q(X)}_1 &= \omega + \norm{(1-p) \Delta + \omega \frac{\Pi_{1}}{d_1}}_1 \\
        &\geq \omega + \abs{(1-p)\norm{\Delta + \omega \frac{\Pi_{1}}{d_1}}_1 - p \norm{\omega \frac{\Pi_{1}}{d_1}}_1 }\\
        &\geq \omega + \abs{(1-p)(1-\omega) - p \omega } 
        = \omega + \abs{1 - p -\omega}
        \\
        &\geq 1 - p.
    \end{align}
\end{proof}

These results are reminiscent of known estimates for Gel'fand numbers for diagonal operators in $\ell_p\to\ell_p$ and their generalization to direct sum operators \cite{Pietsch1980opideals, Ismailov2025}. Those results could potentially be used in this case, but there is a fundamental bottleneck: for direct sum channels, the relevant Banach space is not the direct sum of the Banach spaces, that is, $\HermZero(d_1+d_2)\neq \HermZero(d_1)\oplus\HermZero(d_2)$. This is because one has to add the subspace $\cC$ and a degree of freedom like $\frac{1}{2}(\rho_1\oplus(-\rho_2))$. While $\cC$ is not problematic, the latter term prevents us from extending the lower bound with $l_n$ to $n=d_1^2 + d_2^2 - 1$.

\section{Conclusion}

We have introduced the \emph{contraction and expansion values} $\hkappa_n(T)$ and $\ckappa_n(T)$, two non-increasing sequences indexed by $n=1,\dots,d^2-1$ that refine the trace-distance contraction and expansion coefficients of a quantum channel, which they recover at the extremes $n=1$ and $n=d^2-1$. They are defined through a min--max variational principle that mirrors the Courant--Fischer characterization of the singular values, with the Hilbert--Schmidt norm replaced by the Schatten $1$-norm both in the objective and in the constraint. This replacement gives the sequences a direct operational meaning, which we phrased as two state-discrimination games: one player fixes a set of observables and the other chooses two states that are indistinguishable under those observables, and the resulting optimal guessing probability is governed by $\hkappa_n$ or $\ckappa_n$ according to which player seeks to maximize the output distinguishability.

The structural backbone of these results is the identification of the contraction values with the \emph{Gel'fand numbers} of the channel restricted to traceless Hermitian operators, $T|_0$, within Pietsch's theory of $s$-numbers, and the duality relating the contraction values of an invertible map to the expansion values of its inverse. From these observations we inherit that the sequences are monotone, continuous in the channel, and vanish beyond the algebraic rank. Most importantly for applications, they yield composition bounds showing that the contraction and expansion values of $S\circ T$ are controlled, above and below, by products of the contraction and expansion values of $S$ and $T$. Specialized to $n=1$ these reproduce the familiar bound $\heta(S\circ T)\le \heta(S)\,\heta(T)$, while the accompanying lower bound is, to our knowledge, new and uses the expansion values in an essential way.

We illustrated the framework on three families. For single-qubit channels the $1$-norm coincides with a Euclidean norm, so the two sequences merge and equal the principal semi-axes of the Bloch-sphere image ellipsoid, giving a fully geometric picture. For the $d$-dimensional amplitude damping channel the single-qubit values provide effective estimates; this case is of particular interest because of the recently established, easily checkable criterion for splitting off an amplitude-damping factor, which turns these estimates into bounds for every channel admitting such a decomposition. Finally, for direct-sum channels the contraction coefficient is always trivial, many values vanish, and the remaining ones are estimated by those of the summands.

Several open questions point to natural continuations, all of independent interest to the quantum information community. First, our results bound the contraction \emph{values} but not the \emph{directions} attaining them; a faithful notion of ``contraction vectors'' that composes well seems out of reach beyond the qubit case, where the directions are simply the principal axes of the ellipsoid, and clarifying this obstruction is an open problem. Second, the non-differentiability of the trace norm makes the inner optimizations non-smooth, which places reliable numerical estimation of $\hkappa_n$ and $\ckappa_n$ beyond the present work and motivates dedicated convex-optimization or SDP relaxations. Third, sharper closed-form estimates for elementary channels---notably the $d$-dimensional amplitude damping and direct-sum channels treated here---would improve the composition bounds in practice. Fourth, tensor powers remain elusive: our techniques give no informative control of the values of $T^{\otimes N}$ from those of $T$, a serious gap given the ubiquity of product channels. Beyond the trace distance, extending the construction to other contractive divergences, such as the quantum relative entropy, would align the framework with information-theoretic methods; the unboundedness of the relative entropy, however, introduces genuinely new phenomena, for instance channels whose trace-distance expansion coefficient is strictly positive while the relative-entropy counterpart vanishes. Lastly, applying this framework to infinite-dimensional channels could open new challenges and provide new insights. We expect that progress on any of these fronts will draw further on the bridge between quantum channel contraction and the theory of $s$-numbers established here.

\subsection*{Acknowledgements}

This project was supported by the Basque Government BasQ initiative under the Q-STREAM project. We also aknowledge support from OpenSuperQ+100 (Grant No. \allowbreak{101113946}) of the EU Flagship on Quantum Technologies, from Project Grant No. PID2024-156808NB-I00 and Spanish Ram\'on y Cajal Grant No. RYC-2020-030503-I funded by MI-CIU/AEI/10.13039/501100011033 and by “ERDF A way of making Europe” and “ERDF Invest in your Future”, and from the Spanish Ministry for Digital Transformation and of Civil Service of the Spanish Government through the QUANTUM ENIA project call Quantum Spain, and by the EU through the Recovery, Transformation and Resilience Plan–Next Generation EU within the framework of the Digital Spain 2026 Agenda. RI also acknowledges the support of the Basque Government Ph.D. Grant No. PRE 2021-1-0102.

\bibliography{bibliography_unpruned}{}
\bibliographystyle{plain}

\appendix

\section{Example of extreme relation with the singular values}\label{app:example_singluar_values}

It is known that for non-unital channels at least one singular value of $T$ is strictly larger than $1$ \cite{PerezGarcia2006}. Conveniently, this is not necessary for $T|_0$, for instance, the replacer channel $T\in \cptp(d,d')$ acting as $T(\rho) = \tr(\rho) \, \ketbra{\psi}$ has $s_n(T|_0)=0$, even though $s_1(T) = \sqrt{d}$. However, the following example shows that the dimensional factor in the previous proposition is generally unavoidable:

\begin{example}\label{ex:svs_vs_cvs_evs}
    We take the example from the proof of Theorem III.1 in \cite{PerezGarcia2006}. The theorem demonstrates that for any channel $T\in\cptp(d)$ we have that $s_1(T|_0) \leq \sqrt{\frac{d}{2}}$ or $\sqrt{\frac{d}{2} - \frac{1}{2d}}$ for $d$ even or odd, respectively. Let $\Pi\in \cM_d$ be an orthogonal projector of rank $m=\tr(\Pi)$ and consider the channel $\cE\in\cptp(d)$
    \begin{equation}
        \cE(\rho) = \tr(\Pi \rho) \; \ketbra{0} + \tr(\Pi_{\perp} \rho) \; \ketbra{1},
    \end{equation}
    where $\Pi_{\perp} = \bbI - \Pi$. For simplicity consider $d$ even, then this channel saturates the upper bound $s_1(\cE|_0) = \sqrt{\frac{d}{2}}$, choosing $m=d/2$ and $X = \Pi - \frac{m}{d-m}\Pi_{\perp}$.

    Considering $X\in\HermZero(d)$ with $\tr(\Pi X) = \omega$, observe that $\cE(X) = \omega \, (\ketbra{0} - \ketbra{1})$ implies $s_j(\cE|_0) = \hkappa_j(\cE) = \ckappa_j(\cE) = 0$ for all $j > 1$. Then, choosing states $\rho,\sigma$ with $\Pi \rho = \rho$ and $\Pi \sigma = 0$ and setting $X = \rho - \sigma$ we saturate the ratio $\norm{\cE(X)}_1 / \norm{X}_1 = 1$, which shows $\hkappa_1(\cE) = \ckappa_1(\cE) = 1$. We arrive at
    \begin{equation}
        s_1(\cE|_0) = \sqrt{\frac{d}{2}} \hkappa_1(\cE) = \sqrt{\frac{d}{2}} \ckappa_1(\cE).
    \end{equation}

    It is worth considering the following generalization of the previous example. Consider a orthogonal projector $\Pi'\in\cM_{d'}$ with rank $m'=\tr(\Pi')$ and define the channel $\tilde{\cE} \in \cptp(d, d')$
    \begin{equation}
        \tilde{\cE}(\rho) = \tr(\Pi \rho) \; \frac{\Pi'}{m'} + \tr(\Pi_{\perp} \rho) \; \frac{\Pi'_{\perp}}{m'_{\perp}}.
    \end{equation}
    While the contraction and expansion values are independent of the choice of $m$ and $m'$, it can be shown that
    \begin{equation}
        s_1(\tilde{\cE}|_0) = \sqrt{\frac{1/m' + 1/m'_{\perp}}{1/m + 1/m_{\perp}}},
    \end{equation}
    where $m_{\perp} = d - m$. In particular, taking even dimensions and setting $m=d/2$ and $m'=d'/2$:
    \begin{equation}
        s_1(\tilde{\cE}|_0) = \sqrt{\frac{d}{d'}} \hkappa_1(\tilde{\cE}) = \sqrt{\frac{d}{d'}} \ckappa_1(\tilde{\cE}).
    \end{equation}

    Although we could not establish the tightness of Proposition~\ref{prop:relation_with_singular_values}, these examples show that dimensional factors are necessary.
\end{example}

The example illustrates why the contraction and expansion values are more meaningful than the singular values. It is clear that the channels $\cE$ and $\tilde{\cE}$ transmit equivalent information, as they can be transformed into each other with postprocessing channels associating $\ketbra{0} \leftrightarrow \Pi'/m'$ and $\ketbra{1} \leftrightarrow \Pi'_{\perp}/m'_{\perp}$. While the contraction and expansion values are insensitive to this postprocessing, the singular values vary strongly.

\section{More on \texorpdfstring{$s$}{s}-numbers}\label{app:more_on_s_numbers}

We gather the definitions for six examples of $s$-numbers, their basic relations, and explore their application to define new sequences that characterize channel contraction.

\begin{definition}[Approximation numbers]
    Let $T \in \cL(E, F)$ and $n\in\bbN$.  The $n$th approximation
    number is
    \begin{equation}
        a_n(T):=\inf\big\{ \norm{T-L} : L\in\mathcal L(E, F),\, \rank L< n\big\}.
    \end{equation}
\end{definition}

As the name suggests, $a_n(T)$ quantifies how closely $T$ can be approximated by operators of rank smaller than $n$.

\begin{definition}[Gel'fand numbers]
    Let  $T \in \cL(E, F)$ and $n\in \bbN$. The Gel'fand numbers are given by
    \begin{equation}    
        c_n(T) := \inf \qty{ \norm{T|_W} : W \subset E, \, \codim W < n }.
    \end{equation}
\end{definition}

The definition of the Gel'fand numbers directly mimics the Fischer-Courant minimax principle \cite{Pietsch1987}.

\begin{definition}[Weyl numbers]
    For  $T \in \cL(E, F)$ and $j\in\bbN$ the $n$th Weyl number is defined by
    \begin{equation}
        x_n(T) := \sup\big\{ a_n(T\,S) : S\in\cL (\ell_2,E),\ \norm{S} \leq 1 \big\}.
    \end{equation}
\end{definition}

The Weyl numbers play a significant role in the theory of eigenvalues of operators in Banach spaces \cite{Pietsch1980}.

\begin{definition}[Kolmogorov numbers]
    Let  $T \in \cL(E, F)$ and $n\in\bbN$. The Kolmogorov numbers are
    \begin{equation}
        d_n(T) := \inf\Big\{ \sup_{x\in E, \norm{x}_E \leq 1} \inf_{y\in W} \norm{T(x) - y} : W\subset F,\ \dim W < n\Big\}.
    \end{equation}
\end{definition}

\begin{definition}[Chang numbers]
    For  $T \in \cL(E, F)$ and $n\in\bbN$ the $n$th Chang number is defined by
    \begin{equation}
        y_n(T) := \sup\big\{ a_n(R\,T) : R\in\cL (F,\ell_2),\ \norm{R}\leq 1 \big\}.
    \end{equation}
\end{definition}

\begin{definition}[Hilbert numbers]
    For  $T \in \cL(E, F)$ and $n\in\bbN$ the $n$th Hilbert number is defined by
    \begin{equation}
        h_n(T) :=
        \sup \qty{
            a_n(R\, T\, S) :
            \begin{matrix}
                R\in\cL (F,\ell_2),\ \norm{R}\leq 1 \\
                S\in\cL (\ell_2, E),\ \norm{S} \leq 1
            \end{matrix}
        }
    \end{equation}
\end{definition}

Determining whether any of the established properties of these $s$-numbers apply to the study of quantum channels remains unclear. To close this brief overview, we focus on summarizing their known interrelations depicted in Fig.~\ref{fig:s_numbers_order}. For convenience, each statement is accompanied by a reference to the corresponding result in \cite{Pietsch1987}, which offers a comprehensive review. Readers interested in further details may consult this source directly.

\begin{figure}[tb]
    \centering
    \begin{tikzpicture}
        \tikzstyle{help lines} = [thin, draw = black!50]
        \tikzstyle{block} = [rectangle, draw = black, thick, text width = 2em, align = center, rounded corners, minimum height = 2em]
        \tikzstyle{line} = [draw, thick, shorten >=2pt]
        \tikzstyle{attribute} = [ellipse, draw = black!25, thick, text width = 3em, align = center]
        \tikzstyle{atriangle} = [draw, shorten >=2pt, shorten <=2pt]
    
        \node[block] (app) at (8.50, 3.50) {$a_n$};

        \node[block] (gelfand) at (6.50, 5.50) {$c_n$};
        \node[block] (kolmogorov) at (6.50, 1.50) {$d_n$};

        \draw [-latex, line] (app.north) to [out=90,in=0] (gelfand.east);
        \draw [-latex, line] (app.south) to [out=270,in=0] (kolmogorov.east);

        \node[block] (weyl) at (3.50, 5.50) {$x_n$};
        \node[block] (chang) at (3.50, 1.50) {$y_n$};

        \draw[-latex, line] (gelfand.west) -- (weyl.east);
        \draw[-latex, line] (kolmogorov.west) -- (chang.east);

        \node[block] (hilbert) at (1.50, 3.50) {$h_n$};

        \draw [-latex, line] (weyl.west) to [out=180,in=90] (hilbert.north);
        \draw [-latex, line] (chang.west) to [out=180,in=270] (hilbert.south);

        \draw[open triangle 60 reversed-open triangle 60 reversed, atriangle] (gelfand.south) to (kolmogorov.north);

    \end{tikzpicture}
    \caption{Relations between $s$-numbers. The filled arrows point from larger to smaller $s$-scales and the empty reversed arrows connect dual $s$-scales, in the sense of Proposition~\ref{prop:dual_s_numbers}.}%
    \label{fig:s_numbers_order}
\end{figure}
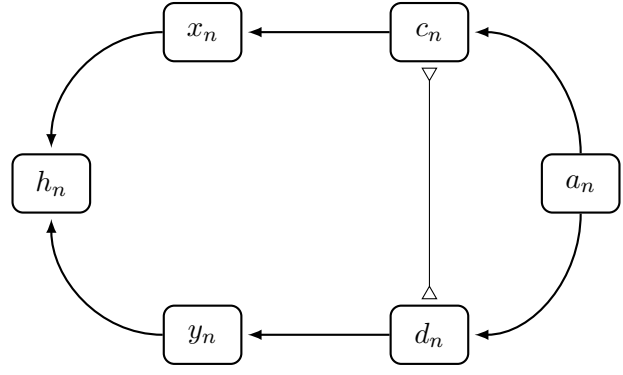

\begin{theorem}[Theorem 2.3.4 in \cite{Pietsch1987}]
    The approximation numbers are the largest $s$-numbers, that is, $s_n(T) \leq a_n(T)$ for any $s$-number sequence $s: T\to (s_n(T))_n$.
\end{theorem}

We say that an $s$-scale $s: T\to (s_n(T))_n$ is injective if it satisfies
\begin{equation}\label{eq:injective_s_number}
    s_n(J_W^F\, T) = s_n(T)
    \quad \text{for all} \quad
    T \in \cL (E, W) \quad 
    \text{and subspace}\quad W \subset F.
\end{equation}

\begin{theorem}[Theorem 4.4 in \cite{Pietsch1974}]
    The Gel'fand numbers are the largest injective $s$-numbers, that is, $s_n(T) \leq c_n(T)$ for any $s$-number sequence $s: T\to (s_n(T))_n$ satisfying Eq.~\eqref{eq:injective_s_number}.
\end{theorem}

\begin{theorem}[Theorem 2.6.3 in \cite{Pietsch1987}]
    The Hilbert numbers are the smallest $s$-numbers, that is, $h_n(T) \leq s_n(T)$ for any $s$-number sequence $s: T\to (s_n(T))_n$.
\end{theorem}

\begin{proposition}[Propositions 2.4.19 \& 2.5.13 in \cite{Pietsch1987}]
    The following relations hold:
    \begin{align}
        x_n(T) \leq c_n(T),
        \quad \text{and} \quad
        y_n(T) \leq d_n(T).
    \end{align}
\end{proposition}

Given a Banach space $E$, its \emph{dual space} $E'$ is defined as the set of all bounded linear functionals on $E$. Each functional $a \in E'$ assigns a scalar value denoted $\langle x, a \rangle$ to every $x \in E$. The dual space $E'$ itself becomes a Banach space endowed with the norm
\begin{equation}
    \norm{a}_{E'} = \sup \{ \abs{\langle x, a\rangle} : x \in E,\, \norm{x}_E \leq 1\}.
\end{equation}
Given $T \in\cL(E, F)$, one can define its dual operator $T' \in \cL(F', E')$. For every $b \in F'$ and $x \in E$, the action of $T'$ is given by
\begin{equation}
    \langle x, T'(b) \rangle = \langle T(x), b \rangle.
\end{equation}
The following is a remarkable relation between the Gel'fand and Kolmogorov numbers:

\begin{proposition}[2.5.6 in \cite{Pietsch1987}]\label{prop:dual_s_numbers}
    The Gel'fand and Kolmogorov numbers are dual to each other:
    \begin{equation}
        c_n(T) = d_n(T')
        \quad \text{and} \quad
        d_n(T) = c_n(T'),
    \end{equation}
    for all compact $T\in\cL(E, F)$.
\end{proposition}

With the application to contraction and expansion values in mind, we identify the direct Banach space with $\HermZero(d)$ equipped with the 1-norm, denote it $E_d \cong (\HermZero(d), \norm{\cdot}_1)$. To identify the dual space, for each $R\in\Herm(d)$ we associate the linear functional $f_R:\Herm(d) \to \bbR$ defined by $f_R(X) = \tr(R X)$, which is bounded $\norm{f_R} = \sup \{ \tr(R X) \, : \, X\in\Herm(d), \norm{X}_1 \leq 1\} = \norm{R}_\infty$ \cite{Bhatia1996}. Moreover, every bounded linear functional on $\Herm(d)$ is of this form, which implies that the dual space of $(\Herm(d), \norm{\cdot}_1)$ can be identified with $(\Herm(d), \norm{\cdot}_\infty)$. We can find the dual $E_d'$ by considering the restriction of these functionals $f_R$ to $\HermZero(d)$. If we decompose $R=c\bbI + Y$ with $c\in\bbR$ and $Y\in\HermZero(d)$, we observe that, as the functional $f_R(X)$ is independent of $c$ when $X\in\HermZero(d)$, we can choose $Y$ as a representative. The dual norm is
\begin{align}\label{eq:dual_hermzero_1norm}
    \norm{f_Y}_{E_d'} &= \sup \{ \tr(Y X) \, : \, X\in\HermZero(d), \norm{X}_1 \leq 1\}  \\
    &= \min_{c\in\bbR} \norm{Y + c \bbI}_{\infty} = \frac{1}{2} (\lambda_{\max}(Y) - \lambda_{\min}(Y)).
\end{align}
The latter follows from combining the upper bound $\norm{f_Y}_{E_d'} = \norm{f_R} \leq \norm{R}_\infty$ and the lower bound $\norm{f_Y}_{E_d'} \geq \frac{1}{2} (\lambda_{\max}(Y) - \lambda_{\min}(Y))$ implied by an appropriate choice of $X$. Therefore, we can identify $E_d'\cong(\HermZero(d), \norm{\cdot}_{1'})$ with the norm $\norm{Y}_{1'} = \min_{c\in\bbR} \norm{Y + c \bbI}_{\infty}$. The dual operator acting on traceless representatives $(T|_0)' \in \cL(E_{d'}', E_d') \cong \cL(\HermZero(d'), \HermZero(d))$, can be related with the Hilbert-Schmidt dual $T^*$, determined by $\tr(A^\dag T(B)) = \tr(T^*(A)^\dag B)$ for all $A \in \cM_{d'}$ and $B \in \cM_{d}$:
\begin{equation}\label{eq:dual_of_T0}
    (T|_0)'(Y) = T^*(Y) - \tr(T^*(Y)) \, \frac{\bbI}{d}.
\end{equation}

Identifying the contraction values with Gel'fand numbers suggests that we could find similar sequences by applying the $s$-numbers to quantum channels and $\hptp_0$ maps in general. For instance, for $T\in\hptp_0(d,d')$ the approximation numbers yield
\begin{equation}\label{eq:cv_from_approximation_numbers}
    a_n(T|_0) =
    \min \qty{
        \norm{T|_0 - L}
        \; :\; L \in \cL(\HermZero(d), \HermZero(d')),
        \, \rank L < n
    }.
\end{equation}
From the relations in the previous section, we know that this sequence is additive and multiplicative (Proposition 2.3.12 in \cite{Pietsch1987}). Therefore, they retain some of the nice properties of the contraction values under composition. Additionally, we know that $\hkappa_n(T) \leq a_n(T|_0)$. The caveat is the lack of a physical interpretation. For $T\in\cptp(d,d')$, if we could guarantee that there is always a $\Phi \in\cptp(d, d')$ such that $L=\Phi|_0$, then the $n$th approximation number would quantify how well a channel $T$ can be approximated by another channel $\Phi$. However, as $L$ is arbitrary, it could happen that there is no such $\Phi$, for instance, if $\norm{L} > 1$. The Weyl, Chang, and Hilbert numbers inherit this lack of interpretability.

The approach is more fruitful if we consider Kolmogorov numbers; a direct substitution yields
\begin{equation}
    d_n(T|_0) = \min_{\dim \cM' \leq n - 1}
                \max_{\norm{X}_1 \leq 1}  \min_{Y \in \cM'}
                    \norm{T(X) - Y}_1,
\end{equation}
where the first minimization is over subspaces $\cM' \subset E_{d'}$. By Proposition~\ref{prop:dual_s_numbers} we can identify $d_n$ with the duals of $c_n$, which by Eq.~\eqref{eq:dual_hermzero_1norm} and \eqref{eq:dual_of_T0} gives
\begin{align}
    d_n(T|_0) &= c_n((T|_0)') \\
    &=
    \min_{\codim \cM' \leq n-1}
    \max_{Y\in \cM', \norm{Y}_{1'} \leq 1}
    \norm{T^*(Y)}_{1'}.
\end{align}
Therefore, the sequence $d_n(T|_0)$ describes the contraction of the spectral diameter, the difference between maximum and minimum eigenvalue, which quantifies the loss in sensitivity of observables. We find that these values are the appropriate duals to study contraction in the Heisenberg picture.

\section{Supplementary material for Section~\ref{sec:amplitude_damping}}

\begin{lemma}[Hadamard channel]\label{lemma:generalized_hadamard}
    Given the positive semidefinite matrix $\Gamma\in\cM_{d}$ with $\Gamma_{ii} = 1$, and a basis $\{\ket{i}\}_{i=1}^d$, the following linear map is CPTP:
    \begin{equation}
        \cP_{\Gamma} : \quad \rho \mapsto \Gamma \odot \rho 
        \quad \text{where} \quad (\Gamma \odot \rho )_{ij} = \Gamma_{ij} \rho_{ij}.
    \end{equation}

    Moreover, its Kraus operators are given by
    \begin{equation}
        E_k = \sqrt{\omega_k} \sum_{i=1}^{d} u_{ik} \ketbra{i},
    \end{equation}
    where the values $\omega_k$ and the unitary $u$ give the spectral decomposition $\Gamma = u \operatorname{diag}(\omega) u^{\dag}$.
\end{lemma}

\begin{proof}
    Since $\Gamma_{i, i}=1$ guarantees that the diagonal is preserved, we clearly have $\tr(\rho) = \tr(\Gamma\circ \rho)$. To prove positivity, introduce $\Gamma$'s spectral decomposition
    \begin{equation}
        \Gamma
        = \sum_{i,j=1}^d \Gamma_{ij} \ketbra{i}{j}
        = \sum_{k=1}^d \omega_k \qty(\sum_{i=1}^d u_{ik} \ket{i}) \qty(\sum_{j=1}^d \bar{u}_{jk} \bra{j})
        = \sum_{k=1}^d E_k \ketbra{s} E_k^\dag,
    \end{equation}
    where $\ket{s} = \sum_{j=1}^d \ket{j}$. Then, observe that $\mel{i}{E_k}{s} = \mel{i}{E_k}{i}$, so we can verify that
    \begin{align}
        \Gamma \circ \rho &= \sum_{i,j=1}^d \Gamma_{ij}\; \rho_{ij} \ketbra{i}{j} 
        = \sum_{i,j=1}^d \qty(\sum_{k=1}^d \mel{i}{E_k}{s} \mel{s}{E_k^\dag}{j} \mel{i}{\rho}{j}) \; \ketbra{i}{j} \\
        &= \sum_{k=1}^d \sum_{i,j=1}^d \mel{i}{E_k}{i} \mel{j}{E_k^\dag}{j} \mel{i}{\rho}{j} \; \ketbra{i}{j} = \sum_{k=1}^d E_k \rho E_k^\dag.
    \end{align}
    Finding a Kraus representation proves that $\cP_\Gamma$ is completely positive. Note that this is only possible if $\omega_k \geq 0$.
\end{proof}

The following is a $d$-dimensional generalization of an existing result \cite{Hirche2024}. It shows that a depolarizing channel can be decomposed into a chain of amplitude damping channels and a Hadamard channel. The idea of the decomposition is depicted in Fig.~\ref{fig:adamp} for $d=3$ and $\lambda=0.7$.

\begin{figure}
    \centering
    \includegraphics[width=0.6\linewidth]{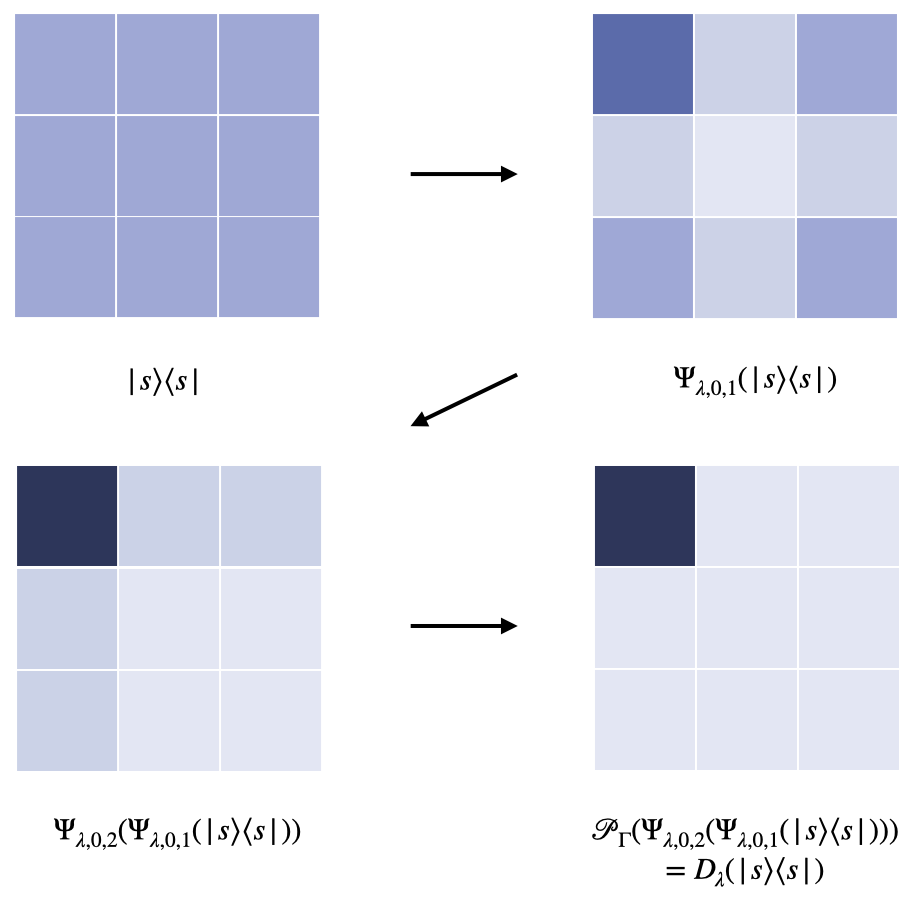}
    \caption{Depicts the decomposition presented in Lemma~\ref{lemma:degrading_dephasing}, which expresses the depolarizing channel as a chain of amplitude damping channels and a Hadamard channel. This example takes $d=3$ and $\lambda=0.7$ and the initial state $\ket{s} = \frac{1}{\sqrt{3}} (\ket{0} + \ket{1} + \ket{2})$.}
    \label{fig:adamp}
\end{figure}

\begin{lemma}[Generalization of Lemma C.2. in \cite{Hirche2024}]\label{lemma:degrading_dephasing}
    Let $\Psi_{\lambda, 0, j}$ be the $d$-dimensional $\ket{j}$ to $\ket{0}$ amplitude damping channel, define
    \begin{equation}
        \cD_{\lambda}(\rho)=(1-\lambda) \rho + \lambda \tr(\rho) \; \ketbra{0},
    \end{equation}
    and take the Hadamard channel $\cP_{\Gamma}$ acting on the canonical basis with
    \begin{equation}
        \Gamma_{kl}
        =
        \begin{cases}
            \sqrt{1-\lambda} & \text{if } k=0 \neq l \text{ or } l=0 \neq k, \\
            1  & \text{otherwise}.
        \end{cases}
    \end{equation}
    
    Then, 
    \begin{equation}
        \cD_{\lambda} = \cP_{\Gamma} \circ \Psi_{\lambda,0,d-1} \circ \cdots \circ \Psi_{\lambda,0,2} \circ \Psi_{\lambda,0,1}
    \end{equation}
\end{lemma}

\begin{proof}
    First, we verify that $\Gamma\geq 0$ to show that the map $\cP_{\Gamma}$ is indeed CPTP by Lemma~\ref{lemma:generalized_hadamard}. Note that
    \begin{equation}
        \Gamma = \ketbra{0} + (d-1) \ketbra{\phi} + \sqrt{1-\lambda} \sqrt{d-1}(\ketbra{0}{\phi} + \ketbra{\phi}{0}),
    \end{equation}
    where we defined $\ket{\phi}=\frac{1}{\sqrt{d-1}} \sum_{k=1}^{d-1} \ket{k}$. Thus, we can represent $\Gamma$ taking $\{\ket{0},\ket{\phi}\}$ as a basis
    \begin{equation}
        \Gamma \dot{=} \pmqty{1 & \sqrt{1-\lambda} \sqrt{d-1} \\ \sqrt{1-\lambda} \sqrt{d-1} & d-1}.
    \end{equation}
    Direct diagonalization shows that
    \begin{equation}
        \Gamma = \lambda_+ \ketbra{\psi_+} + \lambda_- \ketbra{\psi_-}
    \end{equation}
    with
    \begin{equation}
        \lambda_{\pm} = \frac{d}{2}\qty(1 \pm \sqrt{1-4 \lambda \frac{d-1}{d^2}})
        \quad \text{and} \quad
        \ket{\psi_{\pm}} \propto \ket{0} + \frac{\lambda_{\pm}-1}{\sqrt{d-1}\sqrt{1-\lambda}} \ket{\phi}.
    \end{equation}
    Since $\lambda_{\pm}\geq 0$, $\cP_\Gamma$ is completely positive. Additionally, defining $\Pi = \bbI-\ketbra{0}$ note that any state admits a decomposition
    \begin{equation}
        \rho = \ketbra{0} \rho \ketbra{0} + \ketbra{0} \rho \Pi + \Pi \rho \ketbra{0} + \Pi \rho \Pi,
    \end{equation}
    which allows us to act with $\cP_\Gamma$ easily,
    \begin{gather}
        \cP_{\Gamma}(\ketbra{0}) = \ketbra{0}, \;
        \cP_{\Gamma}(\Pi \rho\Pi) = \Pi \rho \Pi,\\
        \cP_{\Gamma}(\ketbra{0} \rho \Pi ) = \sqrt{1-\lambda} \ketbra{0} \rho \Pi, \; \text{and}\;
        \cP_{\Gamma}(\Pi \rho \ketbra{0}) = \sqrt{1-\lambda}\Pi \rho \ketbra{0}.
    \end{gather}

    Introducing $K_{\lambda, 0, k} = \bbI - (1-\sqrt{1-\lambda})\ketbra{k}$ and $L_{\lambda, 0, k} = \sqrt{\lambda} \ketbra{0}{k}$, for $l\neq k$
    \begin{align}
        K_{\lambda, 0, k}K_{\lambda, 0, l} &= K_{\lambda, 0, l}  K_{\lambda, 0, k} = \bbI - (1-\sqrt{1-\lambda})(\ketbra{k}+\ketbra{l}), \\
        K_{\lambda, 0, k} L_{\lambda, 0, l} &= L_{\lambda, 0, l} K_{\lambda, 0, k} = L_{\lambda, 0, l}, \; \text{and} \; L_{\lambda, 0, k} L_{\lambda, 0, l} = 0,
    \end{align}
    hence, denoting $\rho' = \Psi_{\lambda,0,n-1} \circ \cdots \circ \Psi_{\lambda,0,1}(\rho)$,
    \begin{align}
        \rho' &=
        \prod_{k=1}^{d-1} K_{\lambda,0,k} \, \rho\, \prod_{k=1}^{d-1} K_{\lambda,0,k}^{\dag} +  \sum_{k=1}^{d-1} L_{\lambda,0,k} \, \rho \, L_{\lambda, 0, k}^{\dag} \\
        &=
        \qty(\bbI -(1 - \sqrt{1-\lambda})\Pi) \rho \qty(\bbI -(1-\sqrt{1-\lambda}) \Pi)
        +\lambda \tr(\Pi \rho) \, \ketbra{0}\\
        &=
        \qty(\ketbra{0} + \sqrt{1-\lambda}\Pi) \rho \qty(\ketbra{0} + \sqrt{1-\lambda} \Pi)
        +\lambda \tr(\Pi \rho) \, \ketbra{0}\\
        &=
        (\expval{\rho}{0} + \lambda \tr(\Pi \rho)) \ketbra{0}
        + (1-\lambda) \Pi \rho \Pi \\
        & \quad + \sqrt{1-\lambda} (\Pi \rho \ketbra{0} + \ketbra{0} \rho \Pi).
    \end{align}
     Applying $\cP_\Gamma$, we arrive at
    \begin{align}
        \cP_\Gamma ( \rho')
        &=
        (\expval{\rho}{0} + \lambda \tr(\Pi \rho)) \ketbra{0}
        + (1-\lambda) \Pi \rho \Pi \\
        & \quad + (1-\lambda) (\Pi \rho \ketbra{0} + \ketbra{0} \rho \Pi)
        \\
        &=
        \lambda \tr(\rho) \, \ketbra{0}
        +(1-\lambda) \rho 
        \equiv 
        \cD_{\lambda, 0}(\rho).
    \end{align}
\end{proof}

\section{Supplementary material for Section~\ref{sec:direct_sum_channels}}\label{sec:supp_mat_direct_sum}

The extreme values of a direct sum channel are $1$ and $0$. For the intermediate values we need to establish first the following claim.

\begin{lemma}\label{lemma:heta_direct_sum}
    Let $\cV_1 \subseteq \cM_{d_1}$, $\cV_2 \subseteq \cM_{d_2}$, and $\cV_C \subseteq \cC$, with $\cV_1,\cV_2 \neq \{0\}$. Then
    \begin{align}
        \heta(T_1 \oplus T_2, \cV_1 \oplus \cV_2 + \cV_C) 
        = \max \{ \heta(T_1, \cV_1), \heta(T_2, \cV_2) \}, \\
        \ceta(T_1 \oplus T_2, \cV_1 \oplus \cV_2) 
        = \min \{ \ceta(T_1, \cV_1), \ceta(T_2, \cV_2) \}.
    \end{align}
\end{lemma}

\begin{remark}
    Note that the purpose of the requirement $\cV_1,\cV_2\neq \{ 0 \}$ is to guarantee that the subspace expansion and contractions in the left-hand side are defined. However, we can easily let one of them be trivial, for instance, if $\cV_2  = \{ 0 \}$,
    \begin{align}
        \heta(T_1 \oplus T_2, \cV_1 \oplus \{ 0 \} + \cV_C) 
        = \heta(T_1, \cV_1), \\
        \ceta(T_1 \oplus T_2, \cV_1 \oplus \{ 0 \}) 
        = \ceta(T_1, \cV_1).
    \end{align}
\end{remark}

\begin{proof}
    Let $X=X_1\oplus X_2 + X_C$ with $X_1\in\cV_1$, $X_2\in\cV_2$, and $X_C\in\cC$. Since the 1-norm cannot increase when off-diagonal blocks are removed (pinching is 1-norm non-increasing \cite{Bhatia1996}),
    \begin{equation}
        \norm{X}_1 \geq \norm{X_1}_1 + \norm{X_2}_1,
    \end{equation}
    additionally,
    \begin{equation}
        \norm{(T_1\oplus T_2)(X)}_1 =
        \norm{T_1(X_1) \oplus T_2(X_2)}_1 =
        \norm{T_1(X_1)}_1 + \norm{T_2(X_2)}_1.
    \end{equation}
    Hence, $\heta(T_1 \oplus T_2, \cV_1 \oplus \cV_2 + \cV_C)$ is given by
    \begin{align}
        &
        \max \qty{
            \norm{T_1(X_1)}_1 + \norm{T_2(X_2)}_1 \; :\; \norm{X}_1 = 1
        } \\
        &\leq
        \max \qty{
            \norm{T_1(X_1)}_1 + \norm{T_2(X_2)}_1 \; :\; \norm{X_1}_1 + \norm{X_2}_1 \leq 1
        } \\
        &\leq
        \max \qty{
            x\, \heta(T_1, \cV_1) + y\, \heta(T_2, \cV_2) \; :\;  x,y \geq 0,\, x+y\leq 1
        } \\
        &= \max \{ \heta(T_1, \cV_1), \heta(T_2, \cV_2) \}.
    \end{align}
    The proof for the expansion coefficient follows similarly:
    \begin{align}
        \ceta(T_1 \oplus T_2, \cV_1 \oplus \cV_2)
        &=
        \min \qty{
            \norm{T_1(X_1)}_1 + \norm{T_2(X_2)}_1 \; :\; \norm{X_1}_1 + \norm{X_2}_1 = 1
        } \\
        &\geq
        \min \qty{
            x\, \ceta(T_1, \cV_1) + y\, \ceta(T_2, \cV_2) \; :\;  x,y \geq 0,\, x+y = 1
        } \\
        &= \min \{ \ceta(T_1, \cV_1), \ceta(T_2, \cV_2) \}.
    \end{align}
    Both inequalities can be saturated by taking $X_2=0$ and optimizing over $\cV_1$, or symmetrically with $X_1=0$.
\end{proof}

With this we can prove Proposition~\ref{prop:direct_sum_kappas}:

\begin{proof}
    The cases $n=1$ and $n\geq {d_1}^2 + {d_2}^2-1$ are already discussed in this section, so let us focus on the intermediate values. For this, we will prove the existence of subspaces $\cM_n, \cW_n \subseteq \HermZero({d_1}+{d_2})$ such that
    \begin{align}
        \codim \cM_n &= n - 1 \quad \text{and} \quad  \heta(T_1\oplus T_2, \cM_n) = u_{n - 1}, \\
        \dim \cW_n &= n \quad \text{and} \quad  \ceta(T_1\oplus T_2, \cW_n) = l_n,
    \end{align}
    which directly imply Eq.~\eqref{eq:direct_sum_kappas_intermediate}. For this, we use the subspaces
    $\cM^{\alpha}_k$ and $\cW^{\alpha}_k$ with $\alpha=1, 2$ such that
    \begin{align}
       \codim \cM^{\alpha}_k = k - 1, \quad &\text{and} \quad \heta(T_{\alpha}, \cM^{\alpha}_k ) = \hkappa_k(T_{\alpha}) \\
       \dim \cW^{\alpha}_k = k, \quad &\text{and} \quad \ceta(T_{\alpha}, \cW^{\alpha}_k ) = \ckappa_k(T_{\alpha}),
    \end{align}
    which exist by the definition of the contraction and expansion values.

    For a given $n$, assume $u_{n-1} = \hkappa_k(T_1)$ for some $1\leq k\leq n - 1$. This means that in the sequence $(u_{1}, u_{2}, \dots, u_{n - 1})$ there are $k$ values drawn form the contraction values of $T_1$ and $n - 1 - k$ values drawn from the ones of $T_2$. Therefore, if there is an $m$ such that $1 \leq m \leq n-1-k$, we must have $\hkappa_k(T_1) \leq \hkappa_m(T_2)$, which corresponds to the values already included in the sequence. Additionally, if $n-k \leq {d_2}^2-1$, we know that there is at least a contraction value of $T_2$ to be included and that $\hkappa_k(T_1) \geq \hkappa_{n-k}(T_2)$. Define the subspace $\cM_n$ as
    \begin{equation}
        \cM_n = \cC + \cM^{1}_{k} \oplus \cM^{2}_{n - k},
    \end{equation}
    where $1 \leq n - k \leq {d_2}^2$ and we use the temporary convention $\cM_{{d_2}^2} = \{ 0 \}$. By Lemma~\ref{lemma:heta_direct_sum} and the preceding observation, we have that $\heta(T_1\oplus T_2, \cM_n) = \hkappa_k(T_1) = u_{n - 1}$. We complete the proof by verifying that 
    \begin{align}
        \codim\cM_n 
        &=
        \codim \cC - \dim \cM^{1}_{k} - \dim \cM^{2}_{n - k} \\
        &= {d_1}^2 + {d_2}^2 - 1 - \dim \cM^{1}_{k} - \dim \cM^{2}_{n - k} \\
        &= \codim \cM^{1}_{k} + \codim \cM^{2}_{n - k} + 1\\
        &= (k - 1) + (n - k - 1) + 1 = n - 1.
    \end{align}
    The alternative $u_{n - 1} = \hkappa_k(T_2)$ follows equally. As we ultimately have $1\leq k\leq {d_1}^2 - 1$, this strategy works for $2 \leq n \leq {d_2}^2 + {d_1}^2 - 1$.

    Similarly, assume $l_n = \ckappa_k(T_1)$ for some $1\leq k\leq n$. As above, this means that if $n-k \geq 1$, we have $\ckappa_k(T_1) \leq \ckappa_l(T_2)$ for $l \leq n - k$; and that if $n-k+1 \leq {d_2}^2-1$, it must be that $\ckappa_k(T_1) \geq \ckappa_{n-k+1}(T_2)$. Now, let the subspace $\cW_n$ be
    \begin{equation}
        \cW_n = \cW^{1}_{k} \oplus \cW^{2}_{n - k},
    \end{equation}
    where $0 \leq n - k \leq {d_2}^2 - 1$ with the temporary convention $\cW_{0} = \{ 0 \}$. Resorting to Lemma~\ref{lemma:heta_direct_sum}, we have that $\ceta(T_1\oplus T_2, \cW_n) = \ckappa_k(T_1) = l_n$. And
    \begin{equation}
        \dim\cW_n 
        =
        \dim \cW^1_k + \dim \cW^2_{n-k} = n.
    \end{equation}
    The scenario $l_n = \hkappa_k(T_2)$ is equivalent and this is valid for $1 \leq n \leq {d_2}^2 + {d_1}^2 - 2$.
\end{proof}

\end{document}